\documentclass{article}

\usepackage{amsmath}
\usepackage{bm}
\usepackage{algorithm}
\usepackage[noend]{algpseudocode}
\usepackage{amsthm}
\usepackage{amssymb}
\usepackage{graphicx}
\usepackage{subfigure}
\usepackage{soul}
\usepackage[square,sort,comma,numbers]{natbib}

\def\I{\mathbb{I}}
\def\P{\mathbb{P}}
\def\E{\mathbb{E}}
\makeatletter
\def\BState{\State\hskip-\ALG@thistlm}
\makeatother

\newtheorem{theorem}{Theorem}
\newtheorem{lemma}{Lemma}
\newtheorem{remark}{Remark}
\newtheorem{definition}{Definition}

\usepackage[final]{nips_2017}

\usepackage[utf8]{inputenc} % allow utf-8 input
\usepackage[T1]{fontenc}    % use 8-bit T1 fonts
\usepackage{hyperref}       % hyperlinks
\usepackage{url}            % simple URL typesetting
\usepackage{booktabs}       % professional-quality tables
\usepackage{amsfonts}       % blackboard math symbols
\usepackage{nicefrac}       % compact symbols for 1/2, etc.
\usepackage{microtype}      % microtypography
\usepackage{xcolor}

\title{NeuralFDR: Learning Discovery Thresholds \\from Hypothesis Features}

\author{
Fei Xia$^*$,~~~~Martin J. Zhang\thanks{These authors contributed equally to this work and are listed in alphabetical order.}, ~~~~James Zou$^{\dagger}$, ~~~~David Tse\thanks{These authors contributed equally.}\\
  Stanford University\\
  \texttt{\{feixia,jinye,jamesz,dntse\}@stanford.edu} \\
}

\begin{document}
\maketitle

\begin{abstract}
As datasets grow richer, an important challenge is to leverage the full features in the data to maximize the number of useful discoveries while controlling for false positives. We address this problem in the context of multiple hypotheses testing, where for each hypothesis, we observe a p-value along with a set of features specific to that hypothesis. For example, in genetic association studies, each hypothesis tests the correlation between a variant and the trait. We have a rich set of features for each variant (e.g. its location, conservation, epigenetics etc.) which could inform how likely the variant is to have a true association. However popular empirically-validated testing approaches, such as Benjamini-Hochberg's procedure (BH) and independent hypothesis weighting (IHW), either ignore these features or assume that the features are categorical or uni-variate. We propose a new algorithm, \texttt{NeuralFDR}, which automatically learns a discovery threshold as a function of all the hypothesis features. We parametrize the discovery threshold as a neural network, which enables flexible handling of multi-dimensional discrete and continuous features as well as efficient end-to-end optimization. We prove that \texttt{NeuralFDR} has strong false discovery rate (FDR) guarantees, and show that it makes substantially more discoveries in synthetic and real datasets. Moreover, we demonstrate that the learned discovery threshold is directly interpretable.   
\end{abstract}
% \martin{List of things to revise:
% \begin{enumerate}
% \item \st{(Fei) Simulation experiment for the (weak) dependence case (same as rebuttal reviewer 1 point 1)}
% \item \st{Asymptotic FDR control for the (weak) dependent case. (Done, inserted after Theorem 1)}
% \item Justification of using NN
% \item \st{(Fei) details of the algorithm}
% \item \st{(Fei) an experiment (maybe in the appendix) showing the stability of the algorithm (same as rebuttal reviewer 1 point 4)}
% \item \st{(Fei) Report the increases as well as the percentage}
% \item (Perhaps) a short literature review on the FDR control under dependence literatures.
% \item Only have the monotonicity assumption of $f_1(p)$ in the theory.
% \item \st{(Fei) Preparing a python github release}
% \item \st{a discussion on the choice of covariates (relations between $X_i$ and $p_i$). (Done, inserted in the end of section 2)}
% \item a discussion on $\gamma$
% \end{enumerate}}

% \martin{Change of notations: (Now we use $i$ for both hypothesis and the fold number)
% \begin{enumerate}
% \item subscript for hypothesis: $i$, with total number of $n$ hypothesis.
% \item subscript for fold number: $j$, with total number of folds $M$.
% \item fold size: $m = n/M$
% \item subscript for decision rule candidates: $l$, with total number of $L$ candidates. 
% \end{enumerate}
% }

\section{Introduction}
In modern data science, the analyst is often swarmed with a large number of hypotheses --- e.g. is a mutation associated with a certain trait or is this ad effective for that section of the users. Deciding which hypothesis to statistically accept or reject is a ubiquitous task. In standard multiple hypothesis testing, each hypothesis is boiled down to one number, a p-value computed against some null distribution, with a smaller value indicating less likely to be null. We have powerful procedures to systematically reject hypotheses while controlling the false discovery rate (FDR) Note that here the convention is that a ``discovery'' corresponds to a ``rejected'' null hypothesis. 

These FDR procedures are widely used but they ignore additional information that is often available in modern applications. Each hypothesis, in addition to the p-value, could also contain a set of features pertinent to the objects being tested in the hypothesis. In the genetic association setting above, each hypothesis tests whether a mutation is correlated with the trait and we have a p-value for this. Moreover, we also have other features about both the mutation (e.g. its location, epigenetic status, conservation etc.) and the trait (e.g. if the trait is gene expression then we have features on the gene). Together these form a feature representation of the hypothesis. This feature vector is ignored by the standard multiple hypotheses testing procedures. 

In this paper, we present a flexible method using neural networks to learn a nonlinear mapping from hypothesis features to a discovery threshold. Popular procedures for multiple hypotheses testing correspond to having one constant threshold for all the hypotheses (BH \cite{benjamini1995controlling}), or a constant for each group of hypotheses (group BH \cite{hu2010false}, IHW \cite{ignatiadis2016data,ignatiadis2017covariate}). Our algorithm takes account of all the features to automatically learn different thresholds for different hypotheses. Our deep learning architecture enables efficient optimization and gracefully handles both continuous and discrete multi-dimensional hypothesis features. Our theoretical analysis shows that we can control false discovery proportion (FDP) with high probability. We provide extensive simulation on synthetic and real datasets to demonstrate that our algorithm makes more discoveries while controlling FDR compared to state-of-the-art methods. 

\begin{figure} 
\centering
\includegraphics[width=0.8\linewidth]{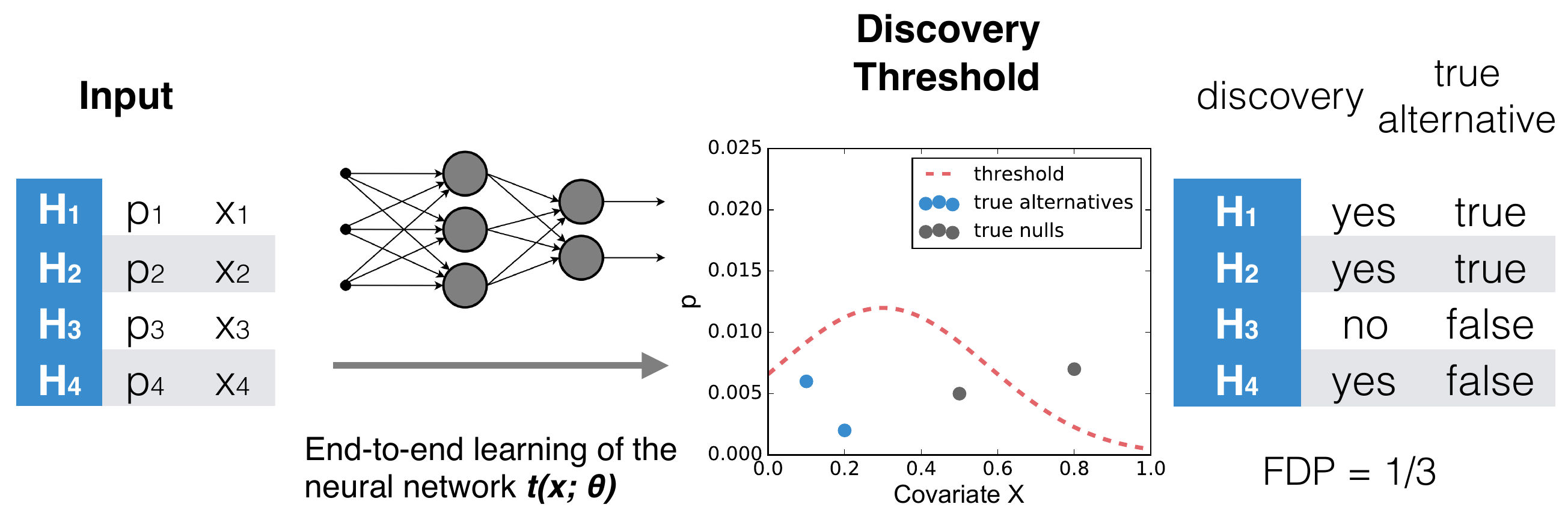}
\caption{NeuralFDR: an end-to-end learning procedure. \label{fig:diag1}}
\end{figure}

\paragraph{Contribution.}
As shown in Fig. \ref{fig:diag1}, we provide \texttt{NeuralFDR}, a practical end-to-end algorithm to the multiple hypotheses testing problem where the hypothesis features can be continuous and multi-dimensional. In contrast, the currently widely-used algorithms either ignore the hypothesis features (BH \cite{benjamini1995controlling}, Storey's BH  \cite{storey2004strong}) or are designed for simple discrete features (group BH \cite{hu2010false}, IHW \cite{ignatiadis2016data}). Our algorithm has several innovative features. We learn a multi-layer perceptron as the discovery threshold and use a \emph{mirroring} technique to robustly estimate false discoveries. We show that \texttt{NeuralFDR} controls false discovery with high probability for independent hypotheses and asymptotically under weak dependence \cite{storey2004strong,hu2010false}, and we demonstrate on both synthetic and real datasets that it controls FDR while making substantially more discoveries. Another advantage of our end-to-end approach is that the learned discovery threshold are directly interpretable. We will illustrate in Sec. \ref{sec:emp} how the threshold conveys biological insights.     

\paragraph{Related works.}
Holm \cite{holm1979simple} investigated the use of p-value weights, where a larger weight suggests that the hypothesis is more likely to be an alternative. 
Benjamini and Hochberg \cite{benjamini1997multiple} considered assigning different losses to different hypotheses according to their importance. 
Some more recent works are \cite{genovese2006false, efron2008simultaneous, hu2010false}. 
In these works, the features are assumed to have some specific forms, either prespecified weights for each hypothesis or the grouping information.  
The more general formulation considered in this paper was purposed quite recently \cite{ignatiadis2016data, li2016multiple, lei2016adapt}. 
It assumes that for each hypothesis, we observe not only a p-value $P_i$ but also a feature $X_i$ lying in some generic space $\mathcal{X}$.
The feature is meant to capture some side information that might bear on the likelihood of a hypothesis to be significant, or on the power of $P_i$ under the alternative, 
but the nature of this relationship is not fully known ahead of time and must be learned from the data.

The recent work most relevant to ours is IHW \cite{ignatiadis2016data}. 
In IHW, the data is grouped into $G$ groups based on the features and the decision threshold is a constant for each group. 
IHW is similar to \texttt{NeuralFDR} in that both methods optimize the parameters of the decision rule to increase the number of discoveries while using cross validation for asymptotic FDR control. 
IHW has several limitations: first, binning the data into $G$ groups can be difficult if the feature space $\mathcal{X}$ is multi-dimensional; 
second, the decision rule, restricted to be a constant for each group, is artificial for continuous features; 
and third, the asymptotic FDR control guarantee requires the number of groups going to infinity, which can be unrealistic. 
In contrast, \texttt{NeuralFDR} uses a neural network to parametrize the decision rule which is much more general and fits the continuous features. 
As demonstrated in the empirical results, it works well with multi-dimensional features. 
In addition to asymptotic FDR control, \texttt{NeuralFDR} also has high-probability false discovery proportion control guarantee with a finite number of hypotheses.

SABHA \cite{li2016multiple} and AdaPT \cite{lei2016adapt} are two recent FDR control frameworks that allow flexible methods to explore the data and compute the feature dependent decision rules. 
The focus there is the framework rather than the end-to-end algorithm as compared to \texttt{NueralFDR}.
For the empirical experiment, SABHA estimates the null proportion using non-parametric methods while AdaPT estimates the distribution of the p-value and the features with a two-group Gamma GLM mixture model and spline regression.
The multi-dimensional case is discussed without empirical validation.
Hence both methods have a similar limitation to IHW in that they do not provide an empirically validated end-to-end approach for multi-dimensional features. 
This issue is addressed in \cite{boca2015regression}, where the null proportion is modeled as a linear combination of some hand-crafted transformation of the features. 
\texttt{NeuralFDR} models this relation in a more flexible way. 

\section{Preliminaries \label{sec:prel}}
We have $n$ hypotheses and each hypothesis $i$ is characterized by a tuple $(P_i, \mathbf{X}_i, H_i)$, where $P_i \in (0,1)$ is the p-value, $\mathbf{X}_i \in \mathcal{X}$ is the hypothesis feature, and $H_i \in \{0,1\}$ indicates if this hypothesis is null ( $H_i=0$) or alternative ( $H_i=1$). 
The p-value $P_i$ represents the probability of observing an equally or more extreme value compared to the testing statistic when the hypothesis is null, and is calculated based on some data different from $\mathbf{X}_i$.
The alternate hypotheses ($H_i = 1$) are the \emph{true signals} that we would like to discover.
A smaller p-value presents stronger evidence for a hypothesis to be alternative. 
In practice, we observe $P_i$ and $\mathbf{X}_i$ but do not know $H_i$. We define the null proportion $\pi_0(\mathbf{x})$ to be the probability that the hypothesis is null conditional on the feature $\mathbf{X}_i=\mathbf{x}$. 
The standard assumption is that under the null ($H_i=0$), the p-value is uniformly distributed in $(0,1)$. Under the alternative ($H_i=1$), we denote the p-value distribution by $f_1(p\vert \mathbf{x})$. In most applications, the p-values under the alternative are systematically smaller than those under the null. A detailed discussion of the assumptions can be found in Sec. \ref{sec:theory}.

The general goal of multiple hypotheses testing is to claim a maximum number of discoveries based on the observations $\{(P_i, \mathbf{X}_i)\}_{i=1}^n$ while controlling the false positives. 
The most popular quantities that conceptualize the false positives are the family-wise error rate (FWER) \cite{dunn1961multiple} and the false discovery rate (FDR) \cite{benjamini1995controlling}.
We specifically consider FDR in this paper.
FDR is the expected proportion of false discoveries, and one closely related quantity, the false discovery proportion (FDP), is the actual proportion of false discoveries.
We note that FDP is the actual realization of FDR. Formally, 
% The most popular quantities that conceptualize the false positives are the false discovery proportion (FDP) and the false discovery rate (FDR) \cite{benjamini1995controlling}. 
\begin{definition} (FDP and FDR)
For any decision rule $t$, let $D(t)$ and $FD(t)$ be the number of discoveries and the number of false discoveries. 
The false discovery proportion $FDP(t)$ and the false discovery rate $FDR(t)$ are defined as $FDP(t) \triangleq FD(t) / D(t)$ and $FDR(t) \triangleq \E[FDP(t)]$. 
\end{definition}
In this paper, we aim to maximize $D(t)$ while controlling $FDP(t)\leq \alpha$ with high probability. 
This is a stronger statement than those in FDR control literature of controlling FDR under the level $\alpha$.

\paragraph{Motivating example.} 
Consider a genetic association study where the genotype and phenotype (e.g. height) are measured in a population. 
Hypothesis $i$ corresponds to testing the correlation between the variant $i$ and the individual's height. 
The null hypothesis is that there is no correlation, and $P_i$ is the probability of observing equally or more extreme values than the empirically observed correlation conditional on the hypothesis is null $H_i=0$.
Small $P_i$ indicates that the null is unlikely. 
Here $H_i = 1$ (or $0$) corresponds to the variant truly is (or is not) associated with height. 
The features $\mathbf{X}_i$ could include the location, conservation, etc. of the variant. 
Note that $\mathbf{X}_i$ is not used to compute $P_i$, but it could contain information about how likely the hypotheses is to be an alternative.
% \martin{The statistical relation between $\mathbf{X}_i$ and $P_i$:}
Careful readers may notice that the distribution of $P_i$ given $\mathbf{X}_i$ is uniform between $0$ and $1$ under the null and $f_1(p\vert\mathbf{x})$ under the alternative, which depends on $\mathbf{x}$.
This implies that $P_i$ and $\mathbf{X}_i$ are \texttt{independent under the null and dependent under the alternative}.

To illustrate why modeling the features could improve discovery power, suppose hypothetically that all the variants truly associated with height reside on a single chromosome $j^*$ and the feature is the chromosome index of each SNP (see Fig. \ref{fig:diag2} (a)). 
Standard multiple testing methods ignore this feature and assign the same discovery threshold to all the chromosomes. 
As there are many purely noisy chromosomes, the p-value threshold must be very small in order to control FDR. 
In contrast, a method that learns the threshold $t(\mathbf{x})$ could learn to assign a higher threshold to chromosome $j^*$ and $0$ to other chromosomes.
As a higher threshold leads to more discoveries and vice versa, this would effectively ignore much of the noise and make more discoveries under the same FDR.

\section{Algorithm Description\label{sec:method}} 
\begin{figure} 
\centering
\subfigure[]{\includegraphics[width=0.3\linewidth]{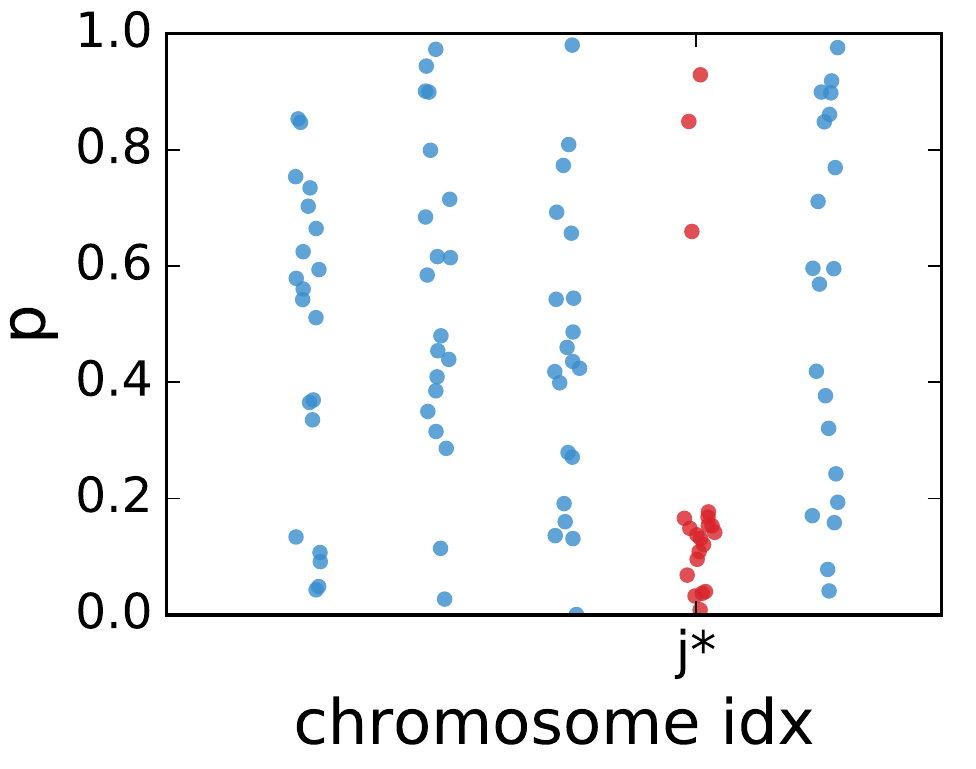}}~~~~~~
\subfigure[]{\includegraphics[width=0.42\linewidth]{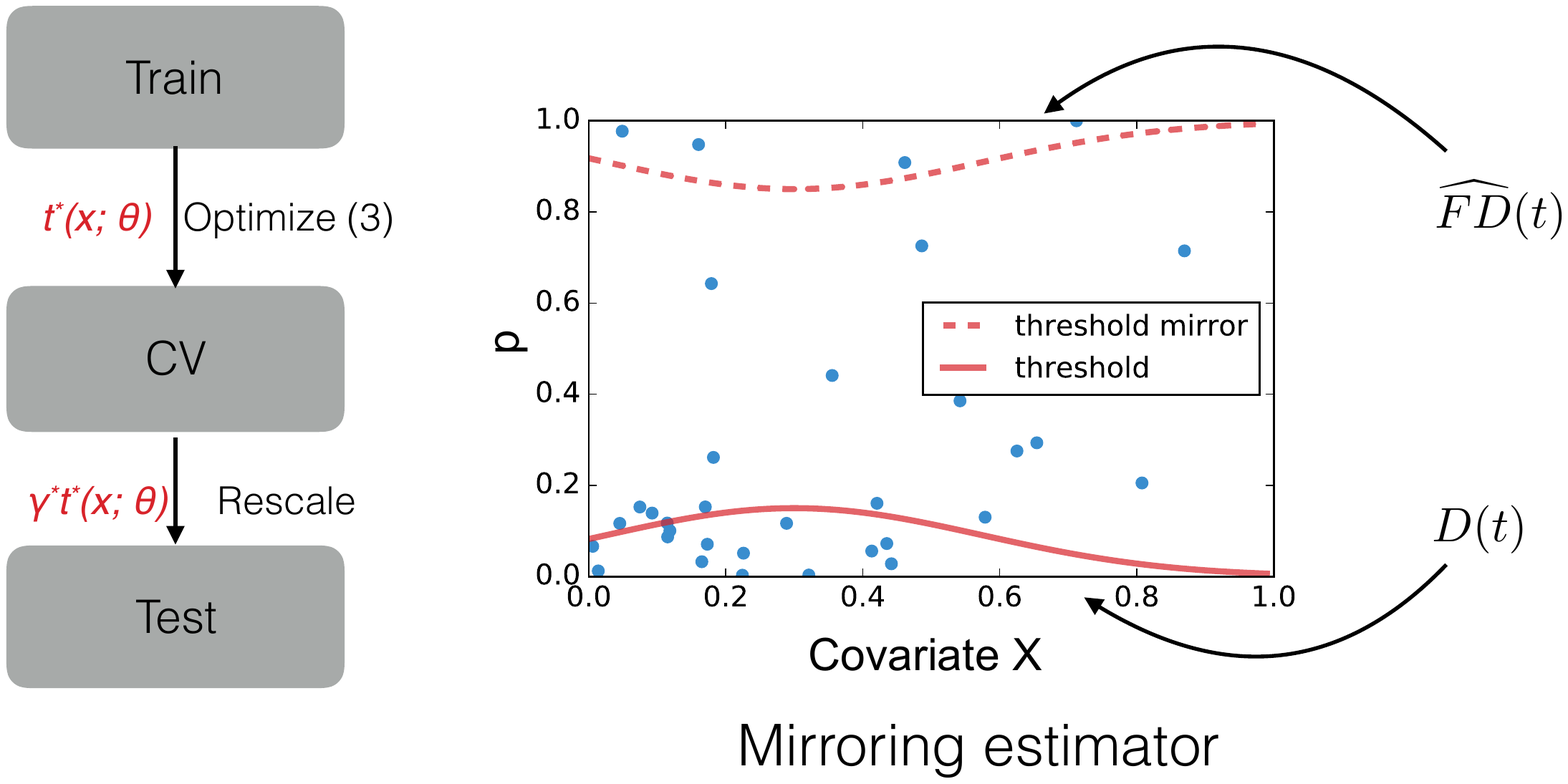}}~~~~~~
\subfigure[]{\includegraphics[width=0.17\linewidth]{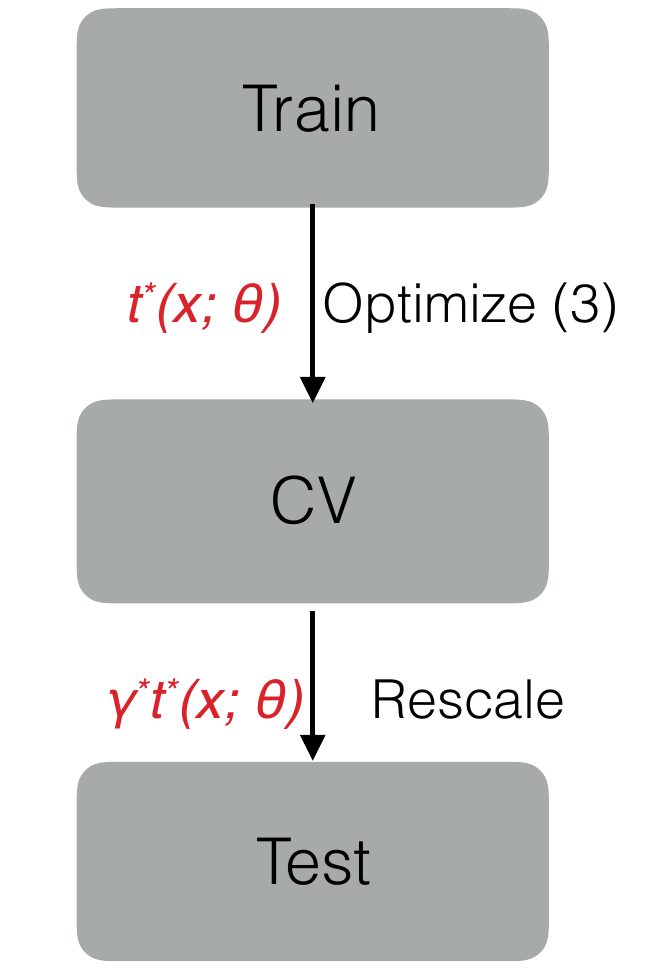}}
\caption{(a) Hypothetical example where small p-values are enriched at chromosome $j^*$. (b) The mirroring estimator. (c) The training and cross validation procedure.  \label{fig:diag2}}
\end{figure}

%For a decision rule $t$, we make the standard assumption that $f_1(p\vert \mathbf{x})$ is non-increasing w.r.t. $p$ for all $\mathbf{x}$. 

Since a smaller p-value presents stronger evidence against the null hypothesis, we consider the threshold decision rule without loss of generality. 
As the null proportion $\pi_0(\mathbf{x})$ and the alternative distribution $f_1(p\vert \mathbf{x})$ vary with $\mathbf{x}$, the threshold should also depend on $\mathbf{x}$.
Therefore, we can write the rule as $t(\mathbf{x})$ in general, which claims hypothesis $i$ to be significant if $P_i < t(\mathbf{X}_i)$. 
Let $\I$ be the indicator function. 
For $t(\mathbf{x})$, the number of discoveries $D(t)$ and the number of false discoveries $FD(t)$ can be expressed as $D(t) = \sum_{i=1}^n \I_{\{P_i < t(\mathbf{X}_i)\}}$ and $FD(t) = \sum_{i=1}^n \I_{\{P_i < t(\mathbf{X}_i), H_i=0\}}$. 
Note that computing $FD(t)$ requires the knowledge of $H_i$, which is not available from the observations. 
Ideally we want to solve $t$ for the following problem:
\begin{align}\label{eq:int_pblm}
\text{maximize}_{t}~ D(t),~~s.t.~FDP(t)\leq \alpha.
\end{align}

Directly solving \eqref{eq:int_pblm} is not possible. 
First, without a parametric representation, $t$ can not be optimized.
Second, while $D(t)$ can be calculated from the data, $FD(t)$ can not, which is needed for evaluating $FDP(t)$.  
Third, while each decision rule candidate $t_j$ controls FDP, optimizing over them may yield a rule that overfits the data and loses FDP control. We next address these three difficulties in order. 

\textbf{First}, the representation of the decision rule $t(\mathbf{x})$ should be flexible enough to address different structures of the data. 
Intuitively, to have maximal discoveries, the landscape of $t(\mathbf{x})$ should be similar to that of the alternative proportion $\pi_1(\mathbf{x})$:  $t(\mathbf{x})$ is large in places where the alternative hypotheses abound. 
As discussed in detail in Sec. \ref{sec:emp}, two structures of $\pi_1(\mathbf{x})$ are typical in practice. 
The first is bumps at a few locations, and the second is slopes that vary with $\mathbf{x}$. 
Hence the representation should at least be able to address these two structures. 
In addition, the number of parameters needed for the representation should not grow exponentially with the dimensionality of $\mathbf{x}$.
Hence non-parametric models, such as the spline-based methods or the kernel based methods, are infeasible. 
Take kernel density estimation in $5$D as example. If we let the kernel width be $0.1$, each kernel contains on average $0.001\%$ of the data.
Then we need at least a million alternative hypothesis data to have a reasonable estimate of the landscape of $\pi_1(\mathbf{x})$.
In this work, we investigate the idea of modeling $t(\mathbf{x})$ using a multilayer perceptron (MLP), which has a high expressive power and has a number of parameters that does not grow exponentially with the dimensionality of the features. 
As demonstrated in Sec. \ref{sec:emp}, it can efficiently recover the two common structures, bumps and slopes, and yield promising results in all real data experiments.

\textbf{Second}, although $FD(t)$ can not be calculated from the data, if it can be overestimated by some $\widehat{FD}(t)$, then the corresponding estimate of FDP, namely $\widehat{FDP}(t) = \widehat{FD}(t) / D(t)$, is also an overestimate. 
Then if $\widehat{FDP}(t)\leq \alpha$, then $FDP(t)\leq \alpha$, yielding the desired FDP control. 
Moreover, if $\widehat{FD}(t)$ is close to $FD(t)$, the FDP control is tight. 
Conditional on $\mathbf{X}=\mathbf{x}$, the rejection region of $p$, namely $(0, t(\mathbf{x}))$, contains a mixture of nulls and alternatives. 
As the null distribution $\text{Unif}(0,1)$ is symmetrical w.r.t. $p=0.5$ while the alternative distribution $f_1(p\vert \mathbf{x})$ is highly asymmetrical, the mirrored region $(1-t(\textbf{x}), 1)$ will contain roughly the same number of nulls but very few alternatives. 
Then the number of hypothesis in $(t(\textbf{x}), 1)$ can be a proxy of the number of nulls in $(0, t(\mathbf{x}))$. 
This idea is illustrated in Fig. \ref{fig:diag2} (b) and we refer to this estimator as the \emph{mirroring estimator}. 
This estimator is also used in \cite{lei2016power,arias2017distribution,lei2016adapt}.
\begin{definition} (The mirroring estimator)
For any decision rule $t$, let $C(t) = \{(p,\mathbf{x}): p < t(\mathbf{x})\}$ be the rejection region of $t$ over $(P_i,\mathbf{X}_i)$ and let its mirrored region %w.r.t. the line $p=0.5$ 
be $C^M(t) = \{(p,\mathbf{x}): p > 1-t(\mathbf{x})\}$.% (see Fig. \ref{fig:diag2} (b)). 
The mirroring estimator of $FD(t)$ is defined as $\widehat{FD}(t) = \sum_{i} \I_{\{(P_i, X_i) \in C^M(t)\}}$.
\end{definition}
The mirroring estimator overestimates the number of false discoveries in expectation:
\begin{lemma} \label{lm:bia_mir}
(Positive bias of the mirroring estimator)
\begin{align}\label{eq:bia_mir}
\E[\widehat{FD}(t)] - \E[FD(t)]=  \sum_{i=1}^n \P\left[(P_i, \mathbf{X}_i) \in C^M(t), H_i=1\right] \geq 0.
\end{align}
\end{lemma}
\begin{remark}
In practice, $t(\mathbf{x})$ is always very small and $f_1(p\vert \mathbf{x})$ approaches $0$ very fast as $p\to 1$. 
Then for any hypothesis with $(P_i,\mathbf{X}_i) \in C^M(t)$, $P_i$ is very close to $1$ and hence $\P(H_i=1)$ is very small. In other words, the bias in \eqref{eq:bia_mir} is much smaller than $\E[FD(t)]$. Thus the estimator is accurate.  
In addition, $\widehat{FD}(t)$ and $FD(t)$ are both sums of $n$ terms. Under mild conditions, they concentrate well around their means. Thus we should expect that $\widehat{FD}(t)$ approximates $FD(t)$ well most of the times. We make this precise in Sec. \ref{sec:theory} in the form of the high probability FDP control statement. 
\end{remark}

\textbf{Third}, we use cross validation to address the overfitting problem introduced by optimization. 
To be more specific, we divide the data into $M$ folds. 
For fold $j$, the decision rule $t_j(\mathbf{x}; \bm{\theta})$, before applied on fold $j$, is trained and cross validated on the rest of the data. 
The cross validation is done by rescaling the learned threshold $t_j(\mathbf{x})$ by a factor $\gamma_j$ so that the corresponding mirror estimate $\widehat{FDP}$ on the CV set is $\alpha$. 
This will not introduce much of additional overfitting since we are only searching over a scalar $\gamma$. 
The discoveries in all $M$ folds are merged as the final result.
We note here distinct folds correspond to subsets of hypotheses rather than samples used to compute the corresponding p-values.
This procedure is shown in Fig. \ref{fig:diag2} (c). 
The details of the procedure as well as the FDP control property are also presented in Sec. \ref{sec:theory}. 

\begin{algorithm}
\caption{\texttt{NeuralFDR}\label{alg:neuralFDR}}
\begin{algorithmic}[1]
% \Procedure{SIA Test}{}
\State Randomly divide the data $\{(P_i,\mathbf{X}_i)\}_{i=1}^n$ into $M$ folds.
\For{fold $j=1, \cdots, M$} 
\State Let the testing data be fold $j$, the CV data be fold $j'\neq j$, and the training data be the rest.
\State Train $t_j(\mathbf{x};\bm{\theta})$ based on the training data by optimizing 
%\james{make . it clear here that the training and CV is done on rest of the hypotheses other than fold $j$.}
\begin{align} \label{eq:opt_pblm}
\text{maximize}_{\bm{\theta}} ~~D(t(\bm{\theta})) ~~ s.t.~~ \widehat{FDP} (t^*_j (\bm{\theta})) \leq \alpha.
\end{align}
\State Rescale $t^*_j(\mathbf{x}; \bm{\theta})$ by $\gamma^*_j$ so that the estimated FDP on the CV data $\widehat{FDP} (\gamma^*_j t^*_j (\bm{\theta})) = \alpha$.
\State Apply $\gamma^*_j t^*_j (\bm{\theta})$ on the data in fold $j$ (the testing data). 
\EndFor
\State Report the discoveries in all $M$ folds. 
% \EndProcedure
\end{algorithmic}
\end{algorithm}

The proposed method \texttt{NeuralFDR} is summarized as Alg. \ref{alg:neuralFDR}. 
There are two techniques that enabled robust training of the neural network. 
First, to have non-vanishing gradients, the indicator functions in \eqref{eq:opt_pblm} are substituted by sigmoid functions with the intensity parameters automatically chosen based on the dataset. % so that the bandwidth of the sigmoid function is roughly the same as the BH threshold applied on the whole data.  
Second, the training process of the neural network may be unstable if we use random initialization.
Hence, we use an initialization method called the $k$-cluster initialization:
1) use $k$-means clustering to divide the data into $k$ clusters based on the features;
2) compute the optimal threshold for each cluster based on the optimal group threshold condition (\eqref{eq:opt_cond} in Sec. \ref{sec:theory});
3) initialize the neural network by training it to fit a smoothed version of the computed thresholds. See Supp. Sec. \ref{sec:S_imp} for more implementation details.

\section{Empirical Results\label{sec:emp}}
We evaluate our method using both simulated data and two real-world datasets\footnote{We released the software at \url{https://github.com/fxia22/NeuralFDR}}.
The implementation details are in Supp. Sec. \ref{sec:S_imp}. 
We compare \texttt{NeuralFDR} with three other methods: BH procedure (BH) \cite{benjamini1995controlling}, Storey’s BH procedure (SBH) with threshold $\lambda=0.4$ \cite{storey2004strong}, and Independent Hypothesis Weighting (IHW) with number of bins and folds set as default \cite{ignatiadis2016data}.
BH and SBH are two most popular methods without using the hypothesis features and IHW is the state-of-the-art method that utilizes hypothesis features.
For IHW, in the multi-dimensional feature case, $k$-means is used to group the hypotheses. In all experiments, $k$ is set to 20 and the group index is provided to IHW as the hypothesis feature.
Other than the FDR control experiment, we set the nominal FDR level $\alpha=0.1$.

% \begin{figure}
% \centering
% \subfigure[]{\includegraphics[width=0.45\linewidth]{FDR2.pdf}}\quad\quad\quad
% \subfigure[]{\includegraphics[width=0.45\linewidth]{FDR1.pdf}}
% \caption{FDP for (a) DataIHW and (b) 1DGM. Dashed line indicate 45 degrees, which is optimal.
% %The dashed line indicate the $45$ degree line, which is the optimal setting when nominal FDR matches actual FDR. SIA and Storey BW are close to optimal, while BW and IHW are overly conservative and suffers loss of power.
% \label{fig:fdr_exp}}
% \end{figure}

% \begin{table}
% \centering
% \caption{Simulated data: \# of discoveries at FDR = 0.1. }%The entries indicate the number of discoveries made by each method on the dataset at FDR of $0.1$.}
% \label{tab:simu}
% \begin{tabular}{l|l|l|l|l|l|l}
%          & DataIHW&1D GM & 1D slope & 2D GM & 2D slope & 5D GM \\
% \hline
% BH       &  2259&    8266        &   11794    &  9917  & 8473  &  9917 \\
% SBH      &  2651 &    9227       &   13593    &  11334  & 9539  & 11334\\
% IHW      & 5074  &  11172     &   12658    &  12175  & 8758 & 11408 \\
%\texttt{NeuralFDR}&  {\bf 6222} &    {\bf 14899}      &  {\bf 15781}     &  {\bf 18844}  & {\bf 10318} & {\bf 18364}  
% \end{tabular}
% \end{table}

\paragraph{Simulated data.}
We first consider DataIHW, the simulated data in the IHW paper ( Supp. 7.2.2 \cite{ignatiadis2016data}). 
Then, we use our own data that are generated to have two feature structures commonly seen in practice, the bumps and the slopes. 
For the bumps, the alternative proportion $\pi_1(\mathbf{x})$ is generated from a Gaussian mixture (GM) to have a few peaks with abundant alternative hypotheses. 
For the slopes, $\pi_1(\mathbf{x})$ is generated linearly dependent with the features. 
After generating $\pi_1(\mathbf{x})$, the p-values are generated following a beta mixture under the alternative and uniform $(0,1)$ under the null.
We generated the data for both 1D and 2D cases, namely 1DGM, 2DGM, 1Dslope, 2Dslope.
For example, Fig. \ref{fig:threshold_exp} (a) shows the alternative proportion of 2Dslope. 
In addition, for the high dimensional feature scenario, we generated a 5D data, 5DGM, which contains the same alternative proportion as 2DGM with $3$ addition non-informative directions. 

We first examine the FDR control property using DataIHW and 1DGM. 
Knowing the ground truth, we plot the FDP (actual FDR) over different values of the nominal FDR $\alpha$ in Fig. \ref{fig:fdr_exp}. 
For a perfect FDR control, the curve should be along the $45$-degree dashed line.
As we can see, all the methods control FDR. 
\texttt{NeuralFDR} controls FDR accurately while IHW tends to make overly conservative decisions. 
Second, we visualize the learned threshold by both \texttt{NeuralFDR} and IWH. As mentioned in Sec. \ref{sec:method}, to make more discoveries, the learned threshold should roughly have the same shape as $\pi_1(\mathbf{x})$.
The learned thresholds of \texttt{NeuralFDR} and IHW for 2Dslope are shown in Fig. \ref{fig:fdr_exp} (b,c).
As we can see, \texttt{NeuralFDR} well recovers the slope structure while IHW fails to assign the highest threshold to the bottom right block.
IHW is forced to be piecewise constant while \texttt{NeuralFDR} can learn a smooth threshold, better recovering the structure of $\pi_1(\mathbf{x})$. In general, methods that partition the hypotheses into discrete groups would not scale for higher-dimensional features. In Appendix~\ref{sec:S_emp}, we show that \texttt{NeuralFDR} is also able to recover the correct threshold for the Gaussian signal. Finally, we report the total numbers of discoveries in Tab. \ref{tab:simu}.

In addition, we ran an experiment with dependent p-values with the same dependency structure as Sec. 3.2 in \cite{ignatiadis2016data}. We call this dataset DataIHW(WD). The number of discoveries are shown in  Tab. \ref{tab:simu}. \texttt{NeuralFDR} has the actual FDP $9.7\%$ while making more discoveries than SBH and IHW. This empirically shows that \texttt{NeuralFDR} also works for weakly dependent data.

All numbers are averaged over $10$ runs of the same simulation setting. 
We can see that \texttt{NeuralFDR} outperforms IHW in all simulated datasets. Moreover, it outperforms IHW by a large margin multi-dimensional feature settings. 

\begin{figure}
\centering
\subfigure[]{\includegraphics[width=0.45\linewidth]{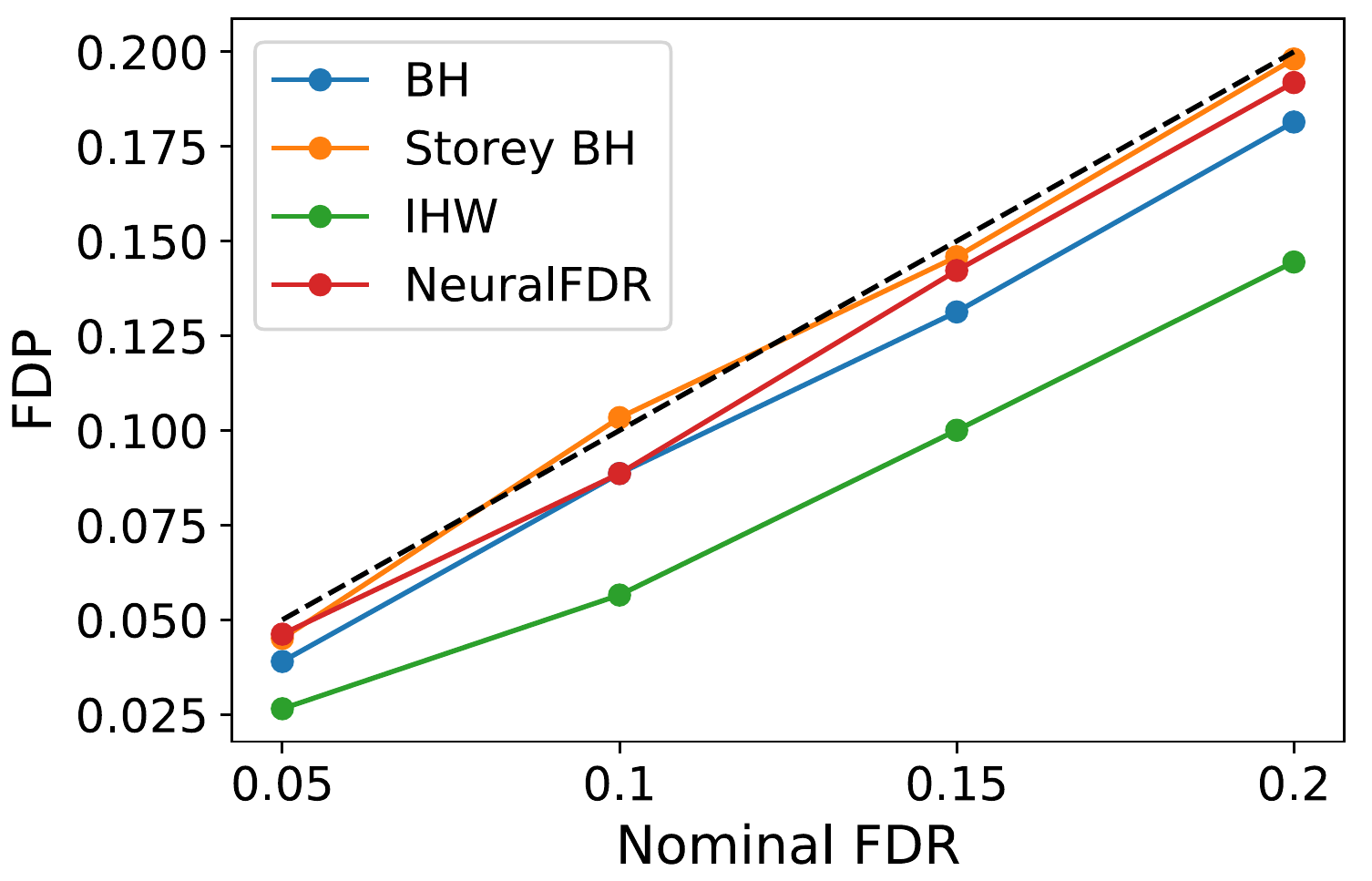}}\quad\quad\quad
\subfigure[]{\includegraphics[width=0.45\linewidth]{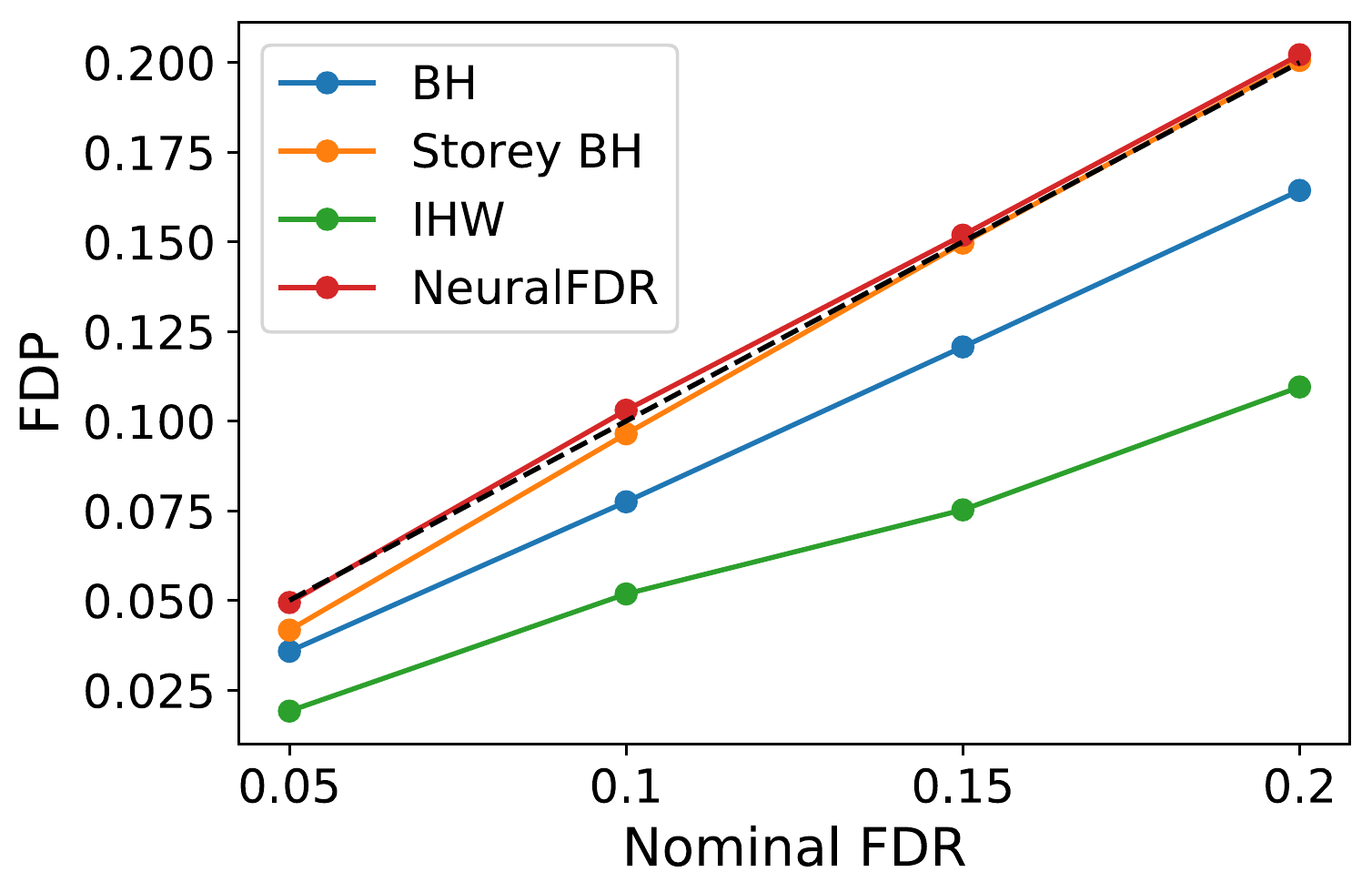}}
\caption{FDP for (a) DataIHW and (b) 1DGM. Dashed line indicate 45 degrees, which is optimal.
%The dashed line indicate the $45$ degree line, which is the optimal setting when nominal FDR matches actual FDR. SIA and Storey BW are close to optimal, while BW and IHW are overly conservative and suffers loss of power.
\label{fig:fdr_exp}}
\end{figure}

\begin{figure}
\centering
\subfigure[Actual alternative proportion for 2Dslope.]{\includegraphics[width=0.3\linewidth]{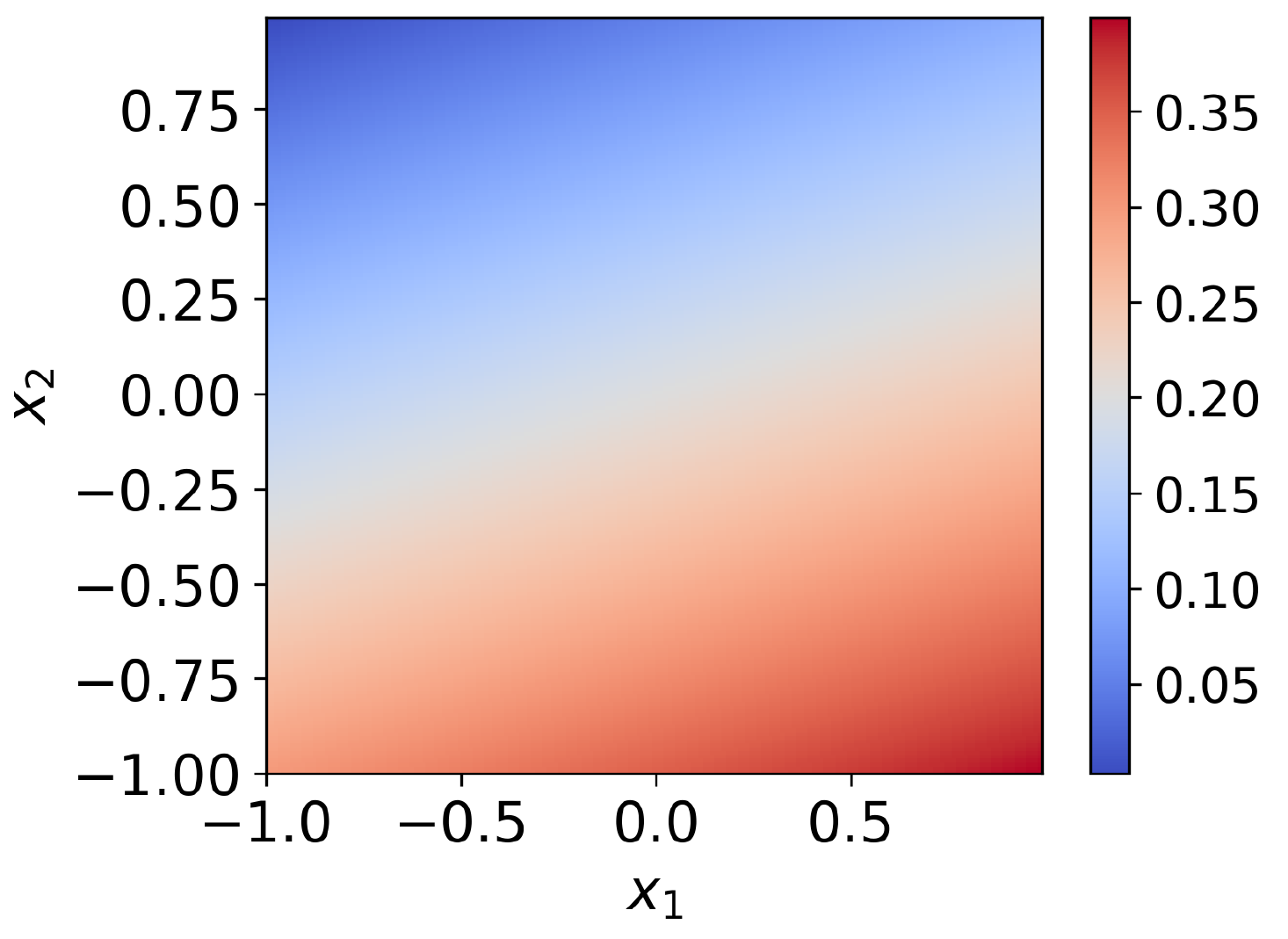}} \quad
\subfigure[\texttt{NeuralFDR}'s learned threshold.]{\includegraphics[width=0.3\linewidth]{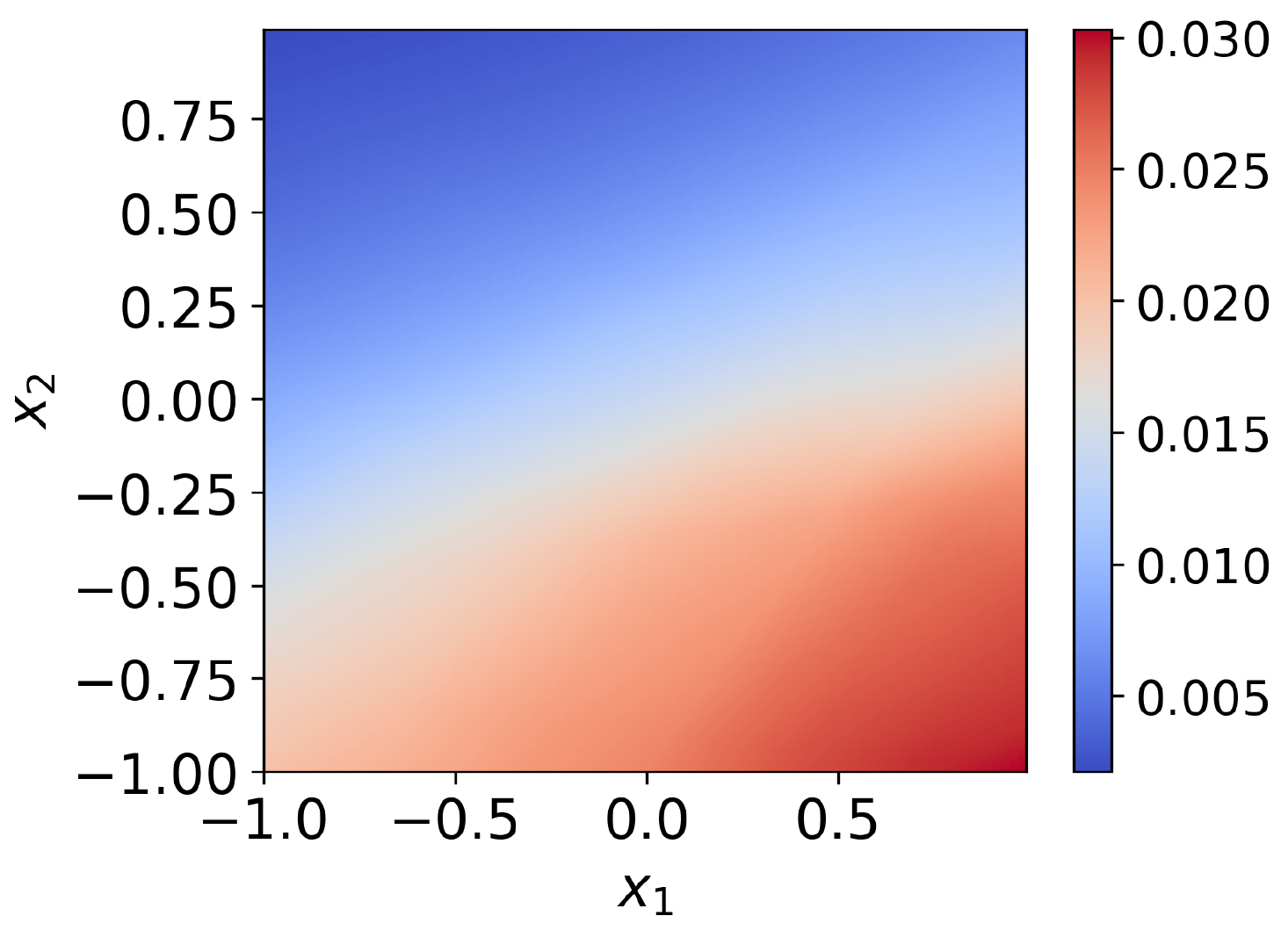}} \quad
\subfigure[IHW's learned threshold]{\includegraphics[width=0.3\linewidth]{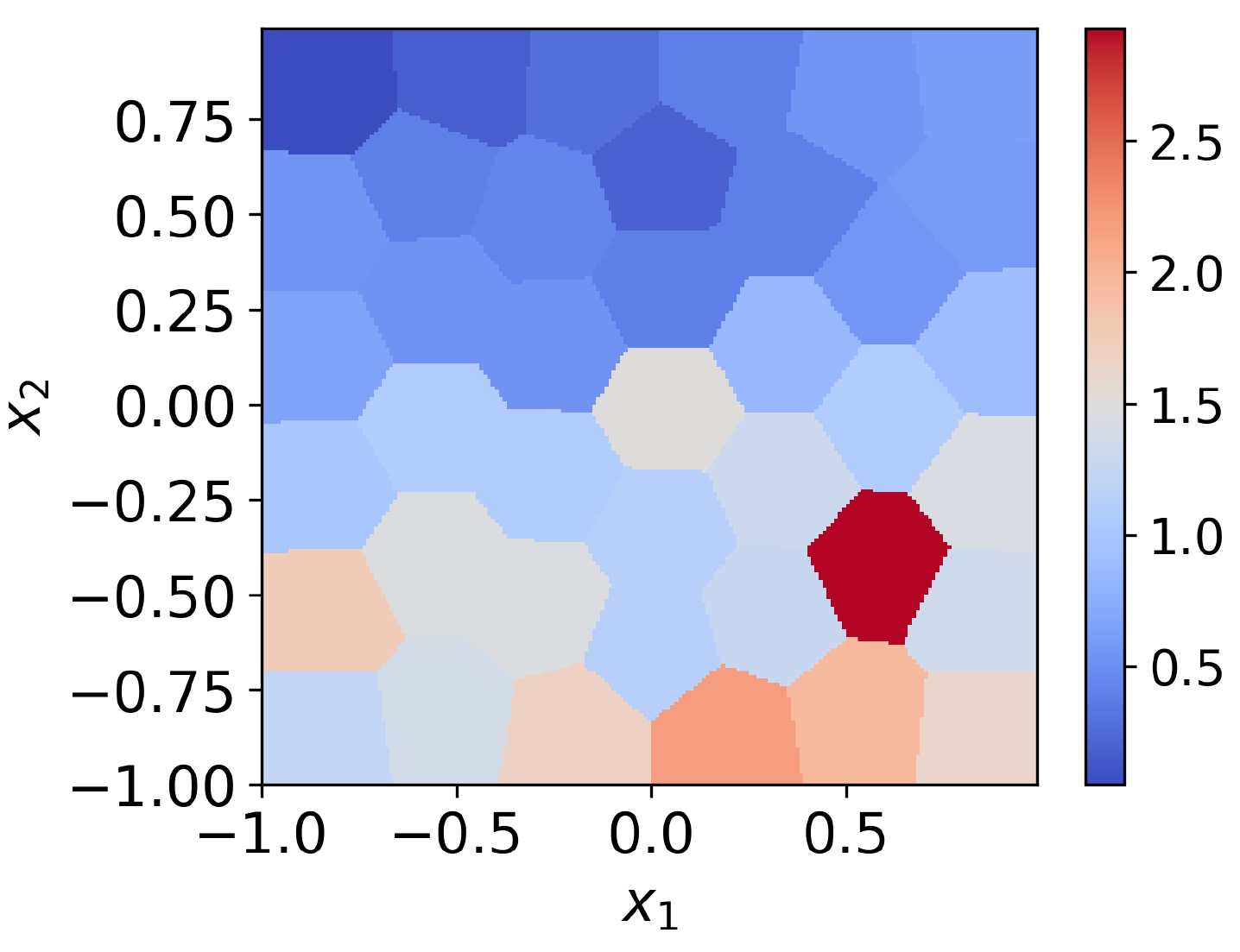}}\\
\subfigure[\texttt{NeuralFDR}'s learned threshold for Airway {\bf log count}.]{\includegraphics[width=0.3\linewidth]{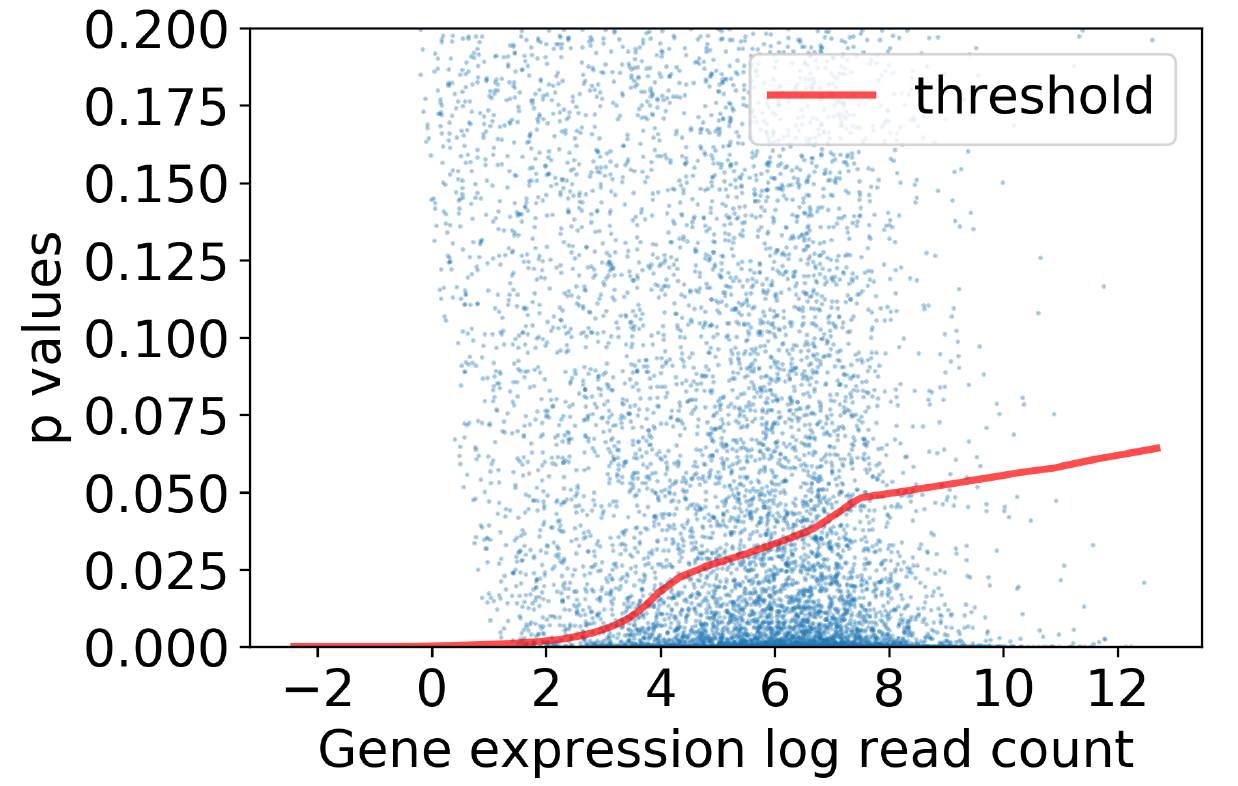}} 
\subfigure[\texttt{NeuralFDR}'s learned threshold for GTEx { \bf log distance}.]{\includegraphics[width=0.3\linewidth]{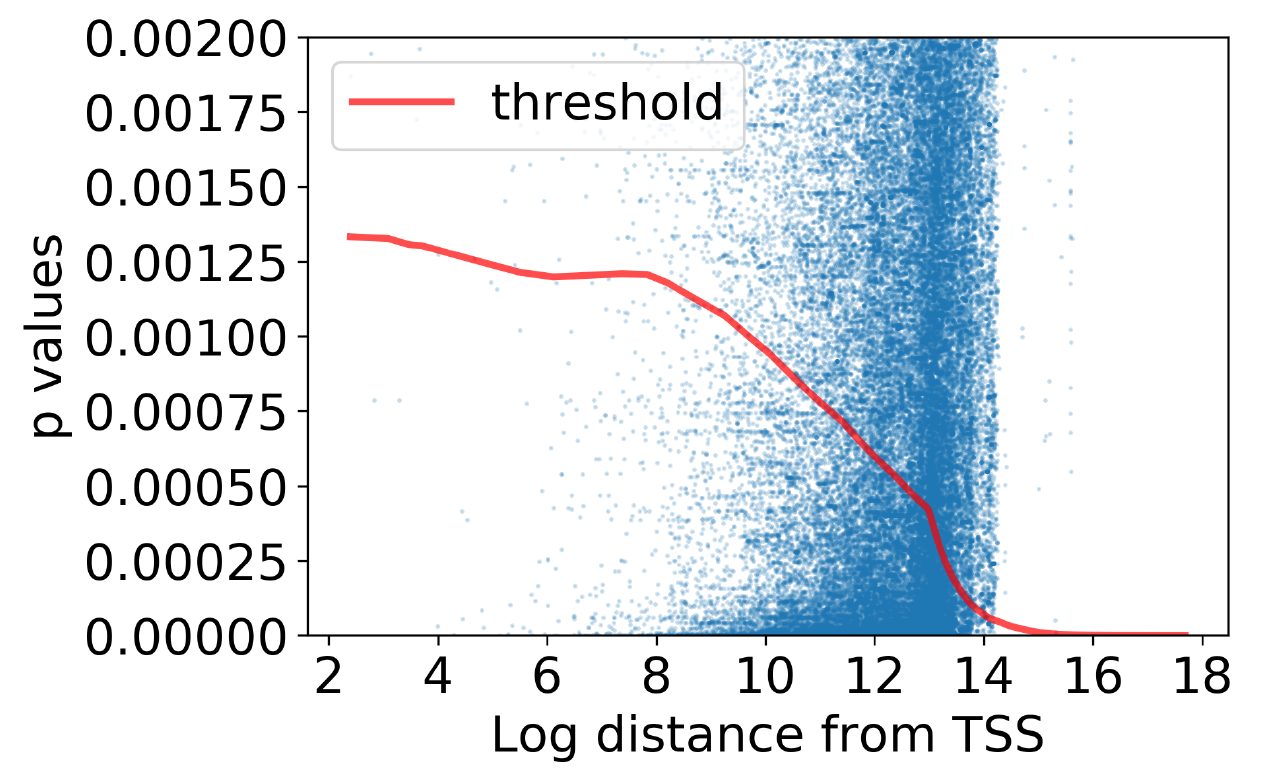}}
\subfigure[\texttt{NeuralFDR}'s learned threshold for GTEx { \bf expression level. }]{\includegraphics[width=0.3\linewidth]{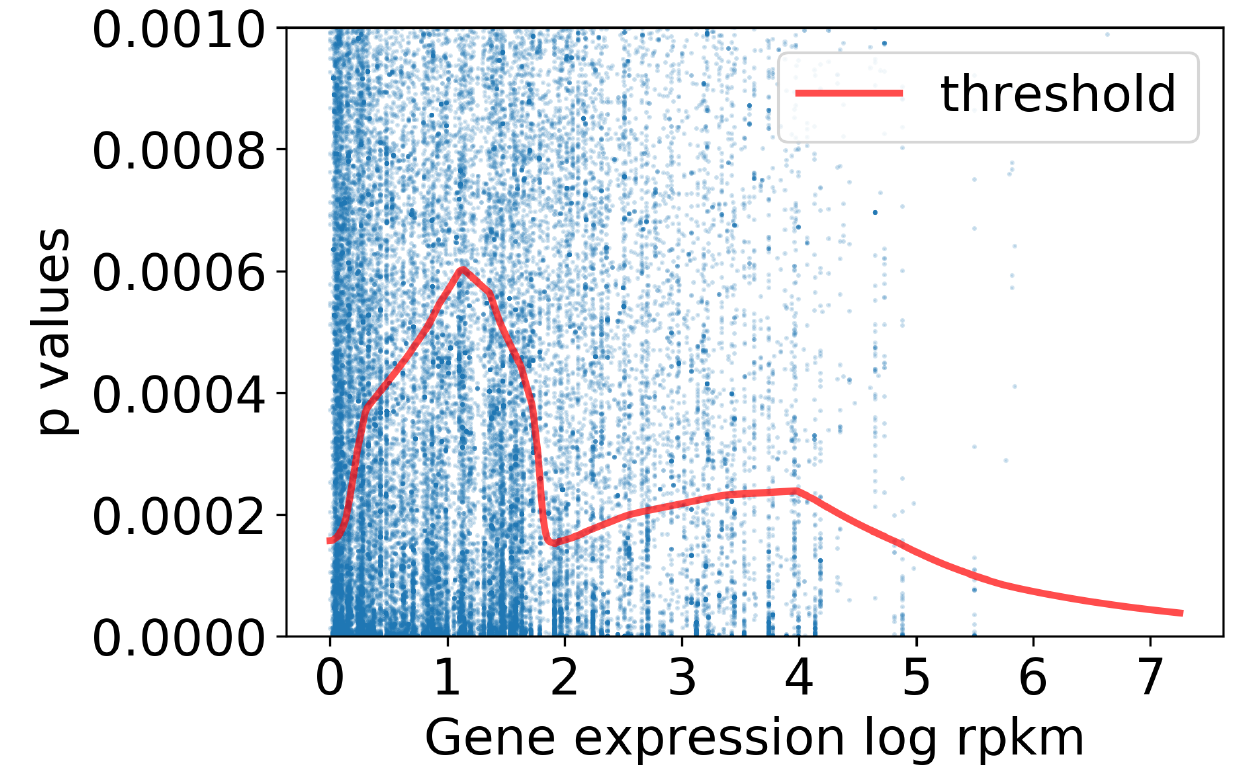}}
\caption{{\bf (a-c)} Results for 2Dslope: (a) the alternative proportion for 2Dslope; (b) \texttt{NeuralFDR}'s learned threshold; (c) IHW's learned threshold.  
{ \bf (d-f)}: Each dot corresponds to one hypothesis. The red curves shows the learned threshold by \texttt{NeuralFDR}:
(d) for {\bf log count} for airway data;
(e) for { \bf log distance} for GTEx data;
(f) for { \bf expression level } for GTEx data. 
%\james{Add title to the panels.}
\label{fig:threshold_exp}}
\end{figure}

\begin{table}
\centering
\caption{Simulated data: \# of discoveries and gain over BH at FDR = 0.1. }%The entries indicate the number of discoveries made by each method on the dataset at FDR of $0.1$.}
\label{tab:simu}
\begin{tabular}{l|l|l|l}
         & DataIHW & DataIHW(WD)&1D GM \\
\hline
BH       &  2259&  6674 &   8266        \\
SBH      &  2651(+17.3\%) & 7844(+17.5\%) &    9227(+11.62\%)  \\
IHW      & 5074(+124.6\%)  & 10382(+55.6\%) & 11172(+35.2\%)    \\
NeuralFDR &  {\bf 6222(+175.4\%)} & {\bf 12153(+82.1\%)} &   {\bf 14899(+80.2\%)}     
\end{tabular}
~
\begin{tabular}{l|l|l|l|l}
         & 1D slope & 2D GM & 2D slope & 5D GM \\
\hline
BH       &    11794    &  9917  & 8473  &  9917 \\
SBH      &   13593(+15.3\%)    &  11334(+14.2\%)  & 9539(+12.58\%)  & 11334(+14.28\%)\\
IHW      &  12658(+7.3\%)    &  12175(+22.7\%)  & 8758(+3.36\%) & 11408(+15.0\%) \\
NeuralFDR &  {\bf 15781(+33.8\%)}     &  {\bf 18844(+90.0\%)}  & {\bf 10318(+21.7\%)} & {\bf 18364(+85.1\%)}  
\end{tabular}

\end{table}

\begin{table}
\centering
\caption{Real data: \# of discoveries at FDR = 0.1.}% The entries indicate the number of discoveries made by each method on the dataset at FDR of $0.1$.}
\label{tab:realdata}
\begin{tabular}{l|l|l|l}
         & Airway & GTEx-dist & GTEx-exp \\
\hline
BH       &      4079        &   29348 & 29348  \\
SBH      &       4038(-1.0\%)       &   29758(+1.4\%)   & 29758(+1.4\%)  \\
IHW      &     4873(+19.5\%)     &   35771(+21.9\%)   & 32195(+9.7\%)  \\
NeuralFDR &     {\bf 6031(+47.9\%)}       &   {\bf 36127(+23.1\%)}  & {\bf 32214(+9.8\%)} 
\end{tabular}

\begin{tabular}{l|l|l|l}
         & GTEx-PhastCons  & GTEx-2D & GTEx-3D\\
\hline
BH       &29348&29348 & 29348   \\
SBH      &29758(+1.4\%) &29758(+1.4\%) & 29758(+1.4\%)\\
IHW      &30241(+3.0\%)& 35705(+21.7\%) & 35598(+21.3\%) \\
NeuralFDR &  {\bf 30525(+4.0\%)} & {\bf 37095(+26.4\%)} & {\bf 37195(+26.7\%)}
\end{tabular}

\end{table}

\paragraph{Airway RNA-Seq data.}
Airway data \cite{himes2014rna} is a RNA-Seq dataset that contains $n=33469$ genes and aims to identify glucocorticoid responsive (GC) genes that modulate cytokine function in airway smooth muscle cells.
The p-values are obtained by a standard two-group differential analysis using DESeq2 \cite{love2014moderated}. 
We consider the log count for each gene as the hypothesis feature. 
As shown in the first column in Tab. \ref{tab:realdata}, \texttt{NeuralFDR} makes $800$ more discoveries than IHW. 
The learned threshold by \texttt{NeuralFDR} is shown in Fig. \ref{fig:threshold_exp} (d).
It increases monotonically with the log count, capturing the positive dependency relation.
Such learned structure is interpretable:
low count genes tend to have higher variances, usually dominating the systematic difference between the two conditions;
on the contrary, it is easier for high counts genes to show a strong signal for differential expression \cite{love2014moderated,ignatiadis2016data}. 

\paragraph{GTEx data.} A major component of the GTEx \cite{gtex2015genotype} study is to quantify expression quantitative trait loci (eQTLs) in human tissues. 
In such an eQTL analysis, each pair of single nucleotide polymorphism (SNP) and nearby gene forms one hypothesis. 
Its p-value is computed under the null hypothesis that the SNP's genotype is not correlated with the gene expression.% (measured in RPKM of RNA-Seq). 
We obtained all the GTEx p-values from chromosome $1$ in a brain tissue (interior caudate), corresponding to $10,623,893$ SNP-gene combinations. 
In the original GTEx eQTL study, no features were considered in the FDR analysis,
%In the original GTEx eQTL study, no features were considered in estimating the p-values discovery threshold during the FDR analysis.
corresponding to running the standard BH or SBH on the p-values.
However, we know many biological features affect whether a SNP is likely to be a true eQTL; i.e. these features could vary the alternative proportion $\pi_1(\mathbf{x})$ and accounting for them could increase the power to discover true eQTL's while guaranteeing that the FDR remains the same. 
For each hypothesis, we generated three features: 1) the distance (GTEx-dist) between the SNP and the gene (measured in log base-pairs) ; 
2) the average expression (GTEx-exp) of the gene across individuals (measured in log rpkm); 
3) the evolutionary conservation measured by the standard PhastCons scores (GTEx-PhastCons). 
% For GTEx-2D, GTEx-dist and GTEx-exp are used.

The numbers of discoveries are shown in Tab. \ref{tab:realdata}. 
For GTEx-2D, GTEx-dist and GTEx-exp are used.
For NeuralFDR, the number of discoveries increases as we put in more and more features, indicating that it can work well with multi-dimensional features. 
For IHW, however, the number of discoveries decreases as more features are incorporated. 
This is because when the feature dimension becomes higher, each bin in IHW will cover a larger space, decreasing the resolution of the piecewise constant function, 
preventing it from capturing the informative part of the feature. 

The learned discovery thresholds of \texttt{NeuralFDR} are directly interpretable and match prior biological knowledge. Fig. \ref{fig:threshold_exp} (e) shows that the threshold is higher when SNP is closer to the gene. This allows more discoveries to be made among nearby SNPs, which is desirable since we know there most of the eQTLs tend to be in cis (i.e. nearby) rather than trans (far away) from the target gene \cite{gtex2015genotype}. Fig. \ref{fig:threshold_exp} (f) shows that the \texttt{NeuralFDR} threshold for gene expression decreases as the gene expression becomes large. This also confirms known biology: the highly expressed genes tend to be more housekeeping genes which are less variable across individuals and hence have fewer eQTLs \cite{gtex2015genotype}. Therefore it is desirable that \texttt{NeuralFDR} learns to place less emphasis on these genes. We also show that \texttt{NeuralFDR} learns to give higher threshold to more conserved variants in Supp. Sec. \ref{sec:S_emp}, which also matches biology.
% \james{Add: In the Appendix, we show that \texttt{NeuralFDR} also learns to give higher threshold to more conserved sited.} 

% \begin{table}
% \centering
% \caption{Simulated data: \# of discoveries at FDR = 0.1. }%The entries indicate the number of discoveries made by each method on the dataset at FDR of $0.1$.}
% \label{tab:simu}
% \begin{tabular}{l|l|l|l|l|l|l}
%          & DataIHW&1D GM & 1D slope & 2D GM & 2D slope & 5D GM \\
% \hline
% BH       &  2259&    8266        &   11794    &  9917  & 8473  &  9917 \\
% SBH      &  2651 &    9227       &   13593    &  11334  & 9539  & 11334\\
% IHW      & 5074  &  11172     &   12658    &  12175  & 8758 & 11408 \\
% \texttt{NeuralFDR} &  {\bf 6222} &    {\bf 14899}      &  {\bf 15781}     &  {\bf 18844}  & {\bf 10318} & {\bf 18364}  
% \end{tabular}
% \end{table}

% \begin{table}
% \centering
% \caption{Real data: \# of discoveries at FDR = 0.1.}% The entries indicate the number of discoveries made by each method on the dataset at FDR of $0.1$.}
% \label{tab:realdata}
% \begin{tabular}{l|l|l|l|l|l|l}
%          & Airway & GTEx-dist & GTEx-exp & GTEx-PhastCons  & GTEx-2D & GTEx-3D\\
% \hline
% BH       &      4079        &   29348 & 29348 &29348&29348 & 29348   \\
% SBH      &       4038       &   29758   & 29758  & 29758 &29758 & 29758\\
% IHW      &     4873     &   35771   & 32195 &30241& 35705 & 35598 \\
% \texttt{NeuralFDR} &     {\bf 6031}       &   {\bf 36127}  & {\bf 32214} &  {\bf 30525} & {\bf 37095} & {\bf 37195}
% \end{tabular}
% %\james{should also show GTEx conservation for completeness.}
% \end{table}

\section{Theoretical Guarantees \label{sec:theory}}
% Formally, following Sec. \ref{sec:prel} and Sec \ref{sec:method},  
We assume the tuples $\{(P_i, \mathbf{X}_i, H_i)\}_{i=1}^n$ are i.i.d. samples from an empirical Bayes model:
\begin{align}\label{eq:model}
\mathbf{X}_i \overset{i.i.d.}{\sim} \mu(\mathbf{X}),~~\left[H_i \vert \mathbf{X}_i = \mathbf{x} \right]\sim \text{Bern}(1-\pi_0(\mathbf{x})),
 \left\{ \begin{array}{ccc}
\left[P_i \vert H_i=0, \mathbf{X}=\mathbf{x}\right] & \sim & \text{Unif}(0,1)\\
\left[P_i \vert H_i=1, \mathbf{X}=\mathbf{x}\right] & \sim & f_1(p\vert \mathbf{x})\\
\end{array}\right.
\end{align}
The features $\mathbf{X}_i$ are drawn i.i.d. from some unknown distribution $\mu(\mathbf{x})$. 
Conditional on the feature $\mathbf{X}_i = \mathbf{x}$, hypothesis $i$ is null with probability $\pi_0 (\mathbf{x})$ and is alternative otherwise. 
The conditional distributions of p-values are $\text{Unif}(0,1)$ under the null and $f_1(p\vert \mathbf{x})$ under the alternative. 
% Recall that we assume that $f_1(p\vert \mathbf{x})$ is monotonically non-increasing. 

\paragraph{FDR control via cross validation.}
The cross validation procedure is described as follows. 
The data is divided randomly into $M$ folds of equal size $m=n/M$.
For fold $j$, let the testing set $\mathcal{D}_{te}(j)$ be itself,
the cross validation set $\mathcal{D}_{cv}(j)$ be any other fold, and the training set $\mathcal{D}_{tr}(j)$ be the remaining.
The size of the three are $m$, $m$, $(M-2)m$ respectively.
For fold $j$, suppose at most $L$ decision rules are calculated based on the training set, namely $t_{j1}, \cdots, t_{jL}$. 
Evaluated on the cross validation set, let $l^*$-th rule be the rule with most discoveries among rules that satisfies 1) its mirroring estimate $\widehat{FDP}(t_{jl})\leq \alpha$; 2) $D(t_{jl}) / m>c_0$, for some small constant $c_0>0$. 
Then, $t_{jl^*}$ is selected to apply on the testing set (fold $j$).
Finally, discoveries from all folds are combined. 

The FDP control follows a standard argument of cross validation. 
Intuitively, the FDP of the rules $\{t_{jl}\}_{l=1}^L$ are estimated based on $\mathcal{D}_{cv}(j)$, a dataset independent of the training set.
Hence there is no overfitting and the overestimation property of the mirroring estimator, as in Lemma \ref{lm:bia_mir}, is statistical valid, leading to a conservative decision that controls FDP. This is formally stated as below. 

\begin{theorem} \label{thrm:FDP_ctrl}
(FDP control) Let $M$ be the number of folds and let $L$ be the maximum number of decision rule candidates evaluated by the cross validation set.
Then with probability at least $1-\beta$, the overall FDP is less than $(1+\Delta)\alpha$, where $\Delta = O \left(\sqrt{\frac{M}{\alpha n} \log \frac{ML}{\beta}}\right)$.
\end{theorem}

\begin{remark}
There are two subtle points.
First, $L$ can not be too large. Otherwise $\mathcal{D}_{cv}(j)$ may eventually be overfitted by being used too many times for FDP estimation.
Second, the FDP estimates may be unstable if the probability of discovery $\E[D(t_{jl}) / m]$ approaches $0$.
Indeed, the mirroring method estimates FDP by $\widehat{FDP}(t_{jl}) = \frac{\widehat{FD}(t_{jl})}{D(t_{jl})}$, where both $\widehat{FD}(t_{jl})$ and $D(t_{jl})$ are i.i.d. sums of $n$ Bernoulli random variables with mean roughly $\alpha \E[D(t_{jl}) / m]$ and $\E[D(t_{jl}) / m]$.
When their means are small, the concentration property will fail. 
So we need $\E[D(t_{jl}) / m]$ to be bounded away from zero. 
Nevertheless this is required in theory but may not be used in practice. 
\end{remark}

\begin{remark} (Asymptotic FDR control under weak dependence)
Besides the i.i.d. case, \texttt{NeuralFDR} can also be extended to control FDR asymptotically under weak dependence \cite{storey2004strong,hu2010false}.
Generalizing the concept in \cite{hu2010false} from discrete groups to continuous features $\mathbf{X}$, the data are under weak dependence if the CDF of $(P_i, X_i)$ for both the null and the alternative proportion converge almost surely to their true values respectively.
The linkage disequilibrium (LD) in GWAS and the correlated genes in RNA-Seq can be addressed by such dependence structure.
In this case, if learned threshold is $c$-Lipschitz continuous for some constant $c$, \texttt{NeuralFDR} will control FDR asymptotically.
The Lipschitz continuity can be achieved, for example, by weight clipping \cite{arjovsky2017wasserstein}, i.e. clamping the weights to a bounded set after each gradient update when training the neural network. See Supp. \ref{supp_sec:FDR_WD} for details. 
\end{remark}

\paragraph{Optimal decision rule with infinite hypotheses.}
When $n=\infty$, we can recover the joint density $f_{P\mathbf{X}}(p,\mathbf{x})$. 
Based on that, the explicit form of the optimal decision rule can be obtained if we are willing to further assumer $f_1(p\vert \mathbf{x})$ is monotonically non-increasing w.r.t. $p$. 
This rule is used for the $k$-cluster initialization for \texttt{NeuralFDR} as mentioned in Sec. \ref{sec:method}. 
Now suppose we know $f_{P\mathbf{X}}(p,\mathbf{x})$.
Then $\mu(\mathbf{x})$ and $f_{P\vert \mathbf{X}}(p\vert \mathbf{x})$  can also be determined. 
Furthermore, as $f_1(p\vert \mathbf{x}) = \frac{1}{1-\pi_0(\mathbf{x})} ( f_{P\vert \mathbf{X}}(p\vert \mathbf{x}) - \pi_0(\mathbf{x}))$, once we specify $\pi_0(\mathbf{x})$, the entire model is specified.
Let $\mathcal{S}(f_{P\mathbf{X}})$ be the set of null proportions $\pi_0(\mathbf{x})$ that produces the model consistent with $f_{P\mathbf{X}}$. 
Because $f_1(p\vert \mathbf{x}) \geq 0$, we have $\forall p, \mathbf{x}, \pi_0(\mathbf{x}) \leq f_{P\vert \mathbf{X}}(p\vert \mathbf{x})$. 
This can be further simplified as $\pi_0(\mathbf{x}) \leq f_{P\vert \mathbf{X}}(1\vert \mathbf{x})$ by recalling that $f_{P\vert \mathbf{X}}(p\vert \mathbf{x})$ is monotonically decreasing w.r.t. $p$. 
Then we know 
\begin{align}
& \mathcal{S}(f_{P\mathbf{X}}) = \{ \pi_0(\mathbf{x}): \forall \mathbf{x}, \pi_0(\mathbf{x}) \leq f_{P\vert \mathbf{X}}(1\vert \mathbf{x})\}.\label{eq:defS}
\end{align}

Given $f_{P\mathbf{X}}(p,\mathbf{x})$, the model is not fully identifiable. Hence we should look for a rule $t$ that maximizes the power while controlling FDP for all elements in $\mathcal{S}(f_{P\mathbf{X}})$.
For $(P_1, \mathbf{X}_1, H_1) \sim (f_{P\mathbf{X}}, \pi_0, f_1)$ following \eqref{eq:model}, the probability of discovery and the probability of false discovery are $P_D(t, f_{P\mathbf{X}})=\P(P_1 \leq t(\mathbf{X}_1))$, $P_{FD}(t, f_{P\mathbf{X}}, \pi_0)=\P(P_1 \leq t(\mathbf{X}_1), H_1=0)$.
Then the FDP is $FDP(t,f_{P\mathbf{X}}, \pi_0) = \frac{P_{FD}(t, f_{P\mathbf{X}}, \pi_0)}{P_D(t, f_{P\mathbf{X}})}$.
In this limiting case, all quantities are deterministic and FDP coincides with FDR. 
Given that the FDP is controlled, maximizing the power is equivalent to maximizing the probability of discovery. Then we have the following minimax problem:
\begin{align}\label{eq:opt}
\max_{t} \min_{ \pi_0 \in\mathcal{S}(f_{P\mathbf{X}})} P_D(t, f_{P\mathbf{X}}) \quad s.t.\quad \max_{\pi_0\in \mathcal{S}(f_{P\mathbf{X}})} FDP(t,f_{P\mathbf{X}}, \pi_0) \leq \alpha,
\end{align}
where $\mathcal{S}(f_{P\mathbf{X}})$ is the set of possible null proportions consistent with $f_{P\mathbf{X}}$, as defined in \eqref{eq:defS}.
% The solution to \eqref{eq:opt}, $t^*(\mathbf{x})$, can be stated as follows.
\begin{theorem}\label{thrm:opt_cond}
Fixing $f_{P\mathbf{X}}$ and let $\pi_0^*(\mathbf{x}) = f_{P\vert \mathbf{X}}(1\vert \mathbf{x})$. If $f_1(p\vert \mathbf{x})$ is monotonically non-increasing w.r.t. $p$, the solution to problem \eqref{eq:opt}, $t^*(\mathbf{x})$, satisfies
\begin{align}\label{eq:opt_cond}
1.\quad\frac{f_{P\mathbf{X}}(1, \mathbf{x})}{f_{P\mathbf{X}}(t^*(\mathbf{x}), \mathbf{x})}=const,~almost~surely~w.r.t.~\mu(\mathbf{x}) \quad \quad 2.\quad FDR(t^*,f_{P\mathbf{X}}, \pi_0^*)= \alpha.
\end{align}
% \begin{align}
% & \frac{f_{P\mathbf{X}}(1, \mathbf{x})}{f_{P\mathbf{X}}(t^*(\mathbf{x}), \mathbf{x})}=const,~a.s.~w.r.t. \mu(\mathbf{x})\label{eq:opt_cond1}\\
% & FDR(t^*,f_{P\mathbf{X}}, \pi_0^*)= \alpha,\label{eq:opt_cond2}
% \end{align}
\end{theorem}

\begin{remark}
To compute the optimal rule $t^*$ by the conditions \eqref{eq:opt_cond}, consider any $t$ that satisfies (\ref{eq:opt_cond}.1). 
According to (\ref{eq:opt_cond}.1), once we specify the value of $t(\mathbf{x})$ at any location $\mathbf{x}$, say $t(0)$, the entire function is determined. 
Also, $FDP(t,f_{P\mathbf{X}}, \pi_0^*)$ is monotonically non-decreasing w.r.t. $t(0)$. 
These suggests the following strategy: starting with $t(0)=0$, keep increasing $t(0)$ until the corresponding FDP equals $\alpha$, which gives us the optimal threshold $t^*$. 
Similar conditions are also mentioned in \cite{lei2016adapt,ignatiadis2016data}.
\end{remark}

\section{Discussion} We proposed \texttt{NeuralFDR}, an end-to-end algorithm to the learn discovery threshold from hypothesis features. 
We showed that the algorithm controls FDR and makes more discoveries on synthetic and real datasets with multi-dimensional features. 
While the results are promising, there are also a few challenges. 
First, we notice that \texttt{NeuralFDR} performs better when both the number of hypotheses and the alternative proportion are large. 
Indeed, in order to have large gradients for the optimization, we need a lot of elements at the decision boundary $t(\mathbf{x})$ and the mirroring boundary $1-t(\mathbf{x})$. 
% Also, to have the high probability FDP control, we need the number of discoveries be moderately large to make the mirroring estimate concentrate well around its mean. 
It is important to improve the performance of \texttt{NeuralFDR} on small datasets with small alternative proportion.
Second, we found that a $10$-layer MLP performed well to model the decision threshold and that shallower networks performed more poorly. 
A better understanding of which network architectures optimally capture signal in the data is also an important question.

% These results are promising, and an important direction is to validate and refine the method on other types of data. 
% Here we found that a ten-layer MLP performed well to model the discovery threshold, and that shallower networks performed more poorly. 
% Better understanding of which network architectures optimally capture signal in the data is also an important question.

\bibliographystyle{plain}
\bibliography{ref}

\begin{thebibliography}{10}

\bibitem{arias2017distribution}
Ery Arias-Castro, Shiyun Chen, et~al.
\newblock Distribution-free multiple testing.
\newblock {\em Electronic Journal of Statistics}, 11(1):1983--2001, 2017.

\bibitem{arjovsky2017wasserstein}
Martin Arjovsky, Soumith Chintala, and L{\'e}on Bottou.
\newblock Wasserstein gan.
\newblock {\em arXiv preprint arXiv:1701.07875}, 2017.

\bibitem{benjamini1995controlling}
Yoav Benjamini and Yosef Hochberg.
\newblock Controlling the false discovery rate: a practical and powerful
  approach to multiple testing.
\newblock {\em Journal of the royal statistical society. Series B
  (Methodological)}, pages 289--300, 1995.

\bibitem{benjamini1997multiple}
Yoav Benjamini and Yosef Hochberg.
\newblock Multiple hypotheses testing with weights.
\newblock {\em Scandinavian Journal of Statistics}, 24(3):407--418, 1997.

\bibitem{boca2015regression}
Simina~M Boca and Jeffrey~T Leek.
\newblock A regression framework for the proportion of true null hypotheses.
\newblock {\em bioRxiv}, page 035675, 2015.

\bibitem{gtex2015genotype}
GTEx Consortium et~al.
\newblock The genotype-tissue expression (gtex) pilot analysis: Multitissue
  gene regulation in humans.
\newblock {\em Science}, 348(6235):648--660, 2015.

\bibitem{duchi2011adaptive}
John Duchi, Elad Hazan, and Yoram Singer.
\newblock Adaptive subgradient methods for online learning and stochastic
  optimization.
\newblock {\em Journal of Machine Learning Research}, 12(Jul):2121--2159, 2011.

\bibitem{dunn1961multiple}
Olive~Jean Dunn.
\newblock Multiple comparisons among means.
\newblock {\em Journal of the American Statistical Association},
  56(293):52--64, 1961.

\bibitem{efron2008simultaneous}
Bradley Efron.
\newblock Simultaneous inference: When should hypothesis testing problems be
  combined?
\newblock {\em The annals of applied statistics}, pages 197--223, 2008.

\bibitem{genovese2006false}
Christopher~R Genovese, Kathryn Roeder, and Larry Wasserman.
\newblock False discovery control with p-value weighting.
\newblock {\em Biometrika}, pages 509--524, 2006.

\bibitem{himes2014rna}
Blanca~E Himes, Xiaofeng Jiang, Peter Wagner, Ruoxi Hu, Qiyu Wang, Barbara
  Klanderman, Reid~M Whitaker, Qingling Duan, Jessica Lasky-Su, Christina
  Nikolos, et~al.
\newblock Rna-seq transcriptome profiling identifies crispld2 as a
  glucocorticoid responsive gene that modulates cytokine function in airway
  smooth muscle cells.
\newblock {\em PloS one}, 9(6):e99625, 2014.

\bibitem{holm1979simple}
Sture Holm.
\newblock A simple sequentially rejective multiple test procedure.
\newblock {\em Scandinavian journal of statistics}, pages 65--70, 1979.

\bibitem{hu2010false}
James~X Hu, Hongyu Zhao, and Harrison~H Zhou.
\newblock False discovery rate control with groups.
\newblock {\em Journal of the American Statistical Association},
  105(491):1215--1227, 2010.

\bibitem{ignatiadis2017covariate}
Nikolaos Ignatiadis and Wolfgang Huber.
\newblock Covariate-powered weighted multiple testing with false discovery rate
  control.
\newblock {\em arXiv preprint arXiv:1701.05179}, 2017.

\bibitem{ignatiadis2016data}
Nikolaos Ignatiadis, Bernd Klaus, Judith~B Zaugg, and Wolfgang Huber.
\newblock Data-driven hypothesis weighting increases detection power in
  genome-scale multiple testing.
\newblock {\em Nature methods}, 13(7):577--580, 2016.

\bibitem{lei2016adapt}
Lihua Lei and William Fithian.
\newblock Adapt: An interactive procedure for multiple testing with side
  information.
\newblock {\em arXiv preprint arXiv:1609.06035}, 2016.

\bibitem{lei2016power}
Lihua Lei and William Fithian.
\newblock Power of ordered hypothesis testing.
\newblock In {\em International Conference on Machine Learning}, pages
  2924--2932, 2016.

\bibitem{li2016multiple}
Ang Li and Rina~Foygel Barber.
\newblock Multiple testing with the structure adaptive benjamini-hochberg
  algorithm.
\newblock {\em arXiv preprint arXiv:1606.07926}, 2016.

\bibitem{love2014moderated}
Michael~I Love, Wolfgang Huber, and Simon Anders.
\newblock Moderated estimation of fold change and dispersion for rna-seq data
  with deseq2.
\newblock {\em Genome biology}, 15(12):550, 2014.

\bibitem{storey2004strong}
John~D Storey, Jonathan~E Taylor, and David Siegmund.
\newblock Strong control, conservative point estimation and simultaneous
  conservative consistency of false discovery rates: a unified approach.
\newblock {\em Journal of the Royal Statistical Society: Series B (Statistical
  Methodology)}, 66(1):187--205, 2004.

\end{thebibliography}

\clearpage
\begin{center}
\textbf{\large Supplemental Materials }
\end{center}
\setcounter{section}{0}

\section{More Discussions on the Empirical Results \label{sec:S_emp}}
In this section, we present more figures of the thresholds learned by \texttt{NeuralFDR} and IHW. 

Fig. \ref{fig:2dgaussian} (a-c) shows the alternative proportion, NeuralFDR's learned threshold, and IHW's learned threshold for the 2D GM simulated data in Sec. \ref{sec:emp}. 
We can see the alternative proportion is well recovered by NeuralFDR. 
To some extent, IHW also recovers the structure but not with high resolution because its threshold is limited to have a constant threshold for each group.
This causes a loss in resolution in informative directions. 

Fig. \ref{fig:2dgaussian} (e,d) shows the learned threshold for the GTEx-2D experiment, where we recall for this experiment, the features distance (GTEx-dist) and expression level (GTEx-exp) are used. 
We can see that \texttt{NeuralFDR} captures the structure that the alternative proportion is large when the distance is small and when the expression level is small. 
This matches the biological explanation as illustrated in Sec \ref{sec:emp}.
However, IHW does not capture such structure very well. 

Fig. \ref{fig:2dgaussian} (f) shows the learned threshold for the GTEx-PhastCons experiment. The threshold is higher for more conserved regions but the difference is not very significant, showing that this covariate contains less information than distance (GTEx-dist) and expression (GTEx-exp). This is consistent with the observation that both IHW and \texttt{NeuralFDR} make fewer discoveries with PhastCons score than with distance or expression.

\begin{figure}[htb]
\centering
\subfigure[]{\includegraphics[width=0.32\linewidth]{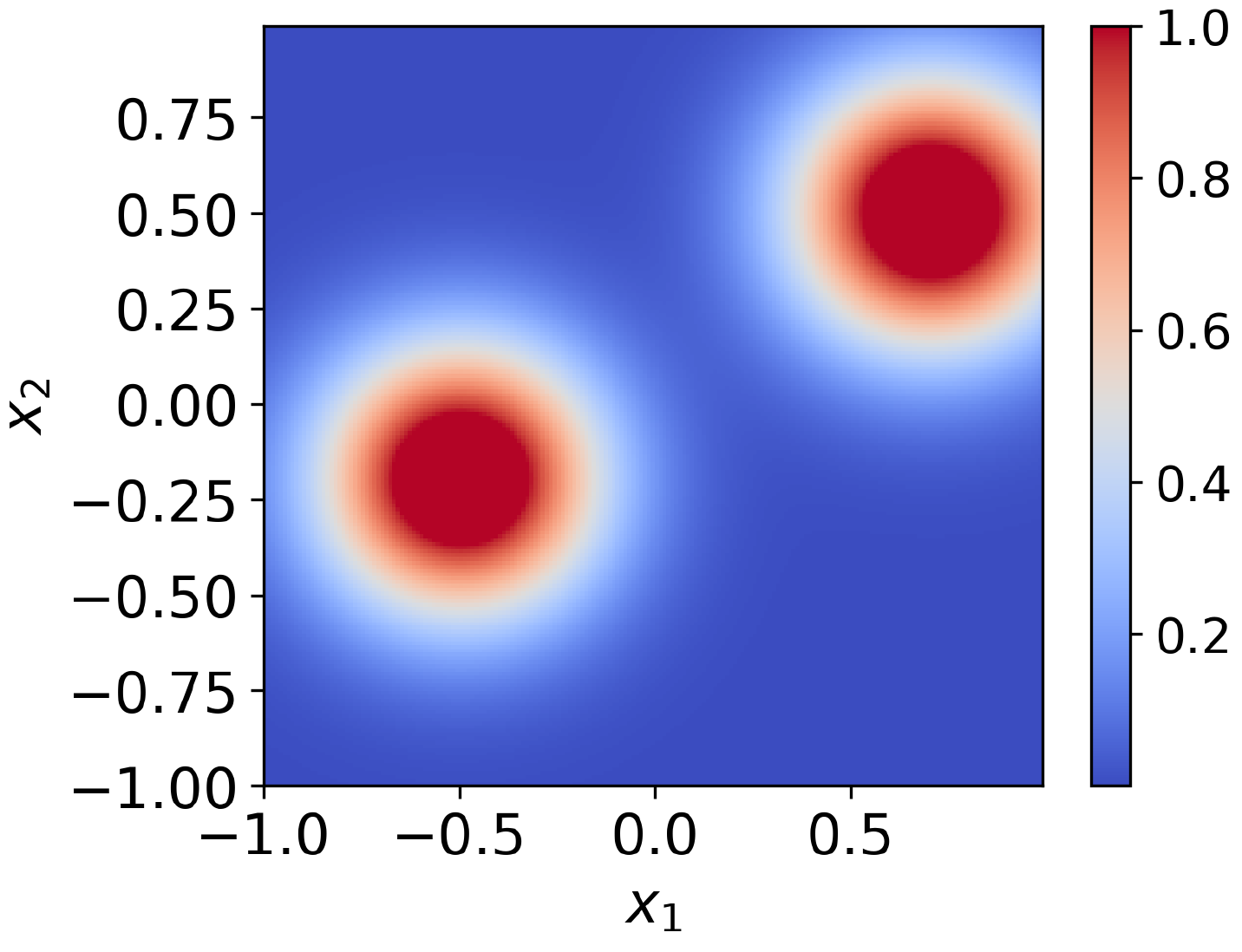}}
\subfigure[]{\includegraphics[width=0.32\linewidth]{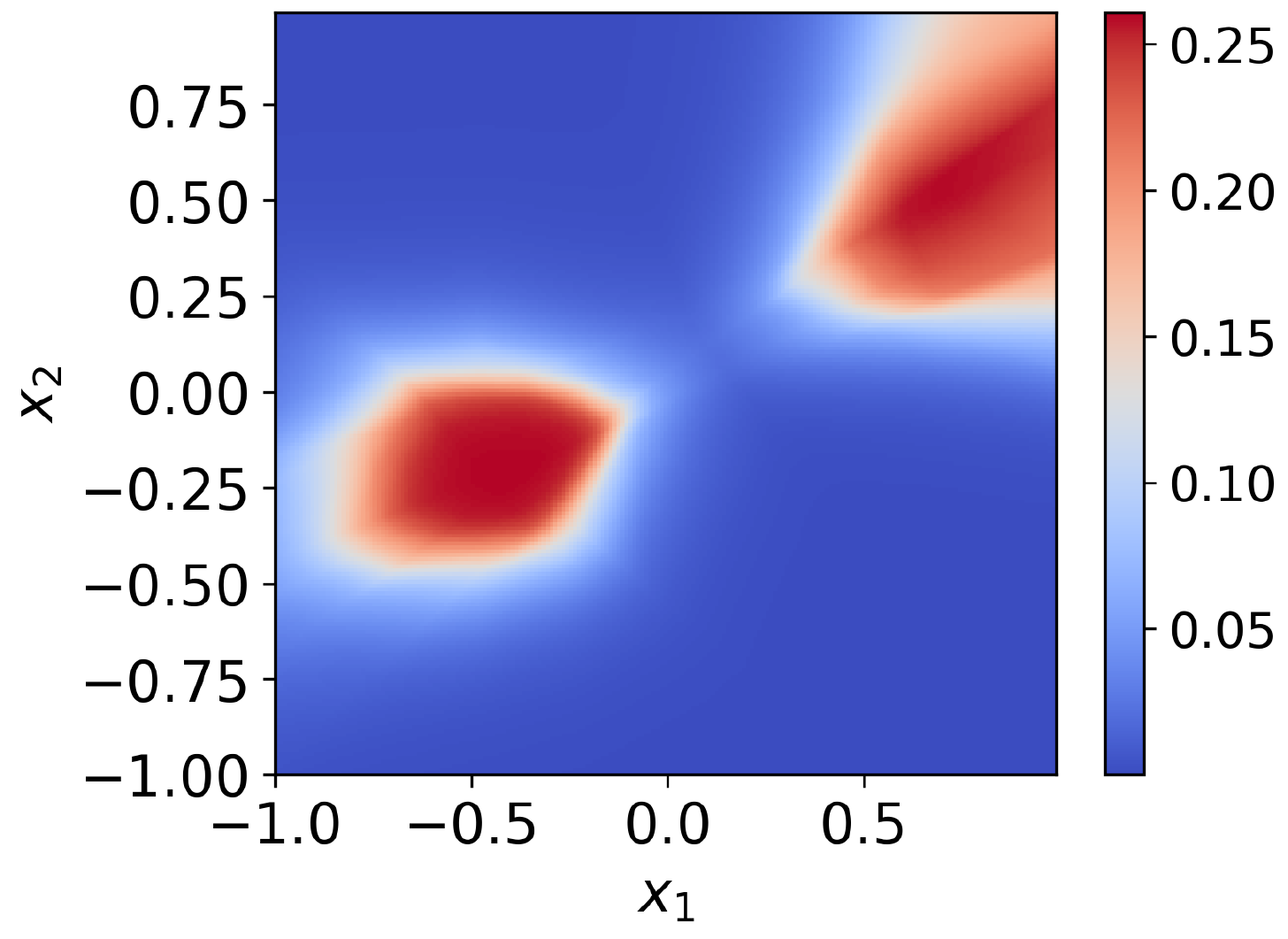}}
\subfigure[]{\includegraphics[width=0.32\linewidth]{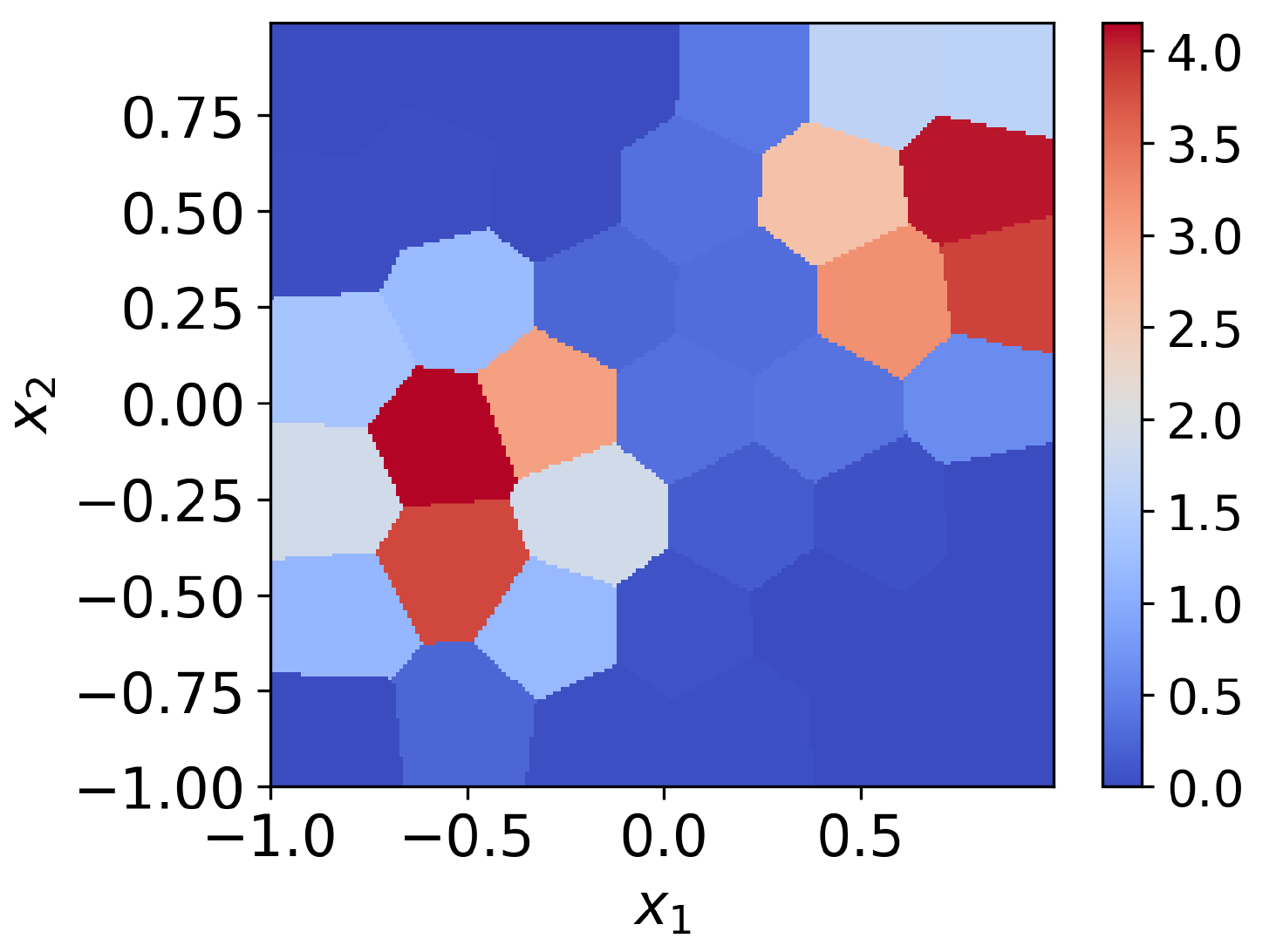}} \\
\subfigure[]{\includegraphics[width=0.33\linewidth]{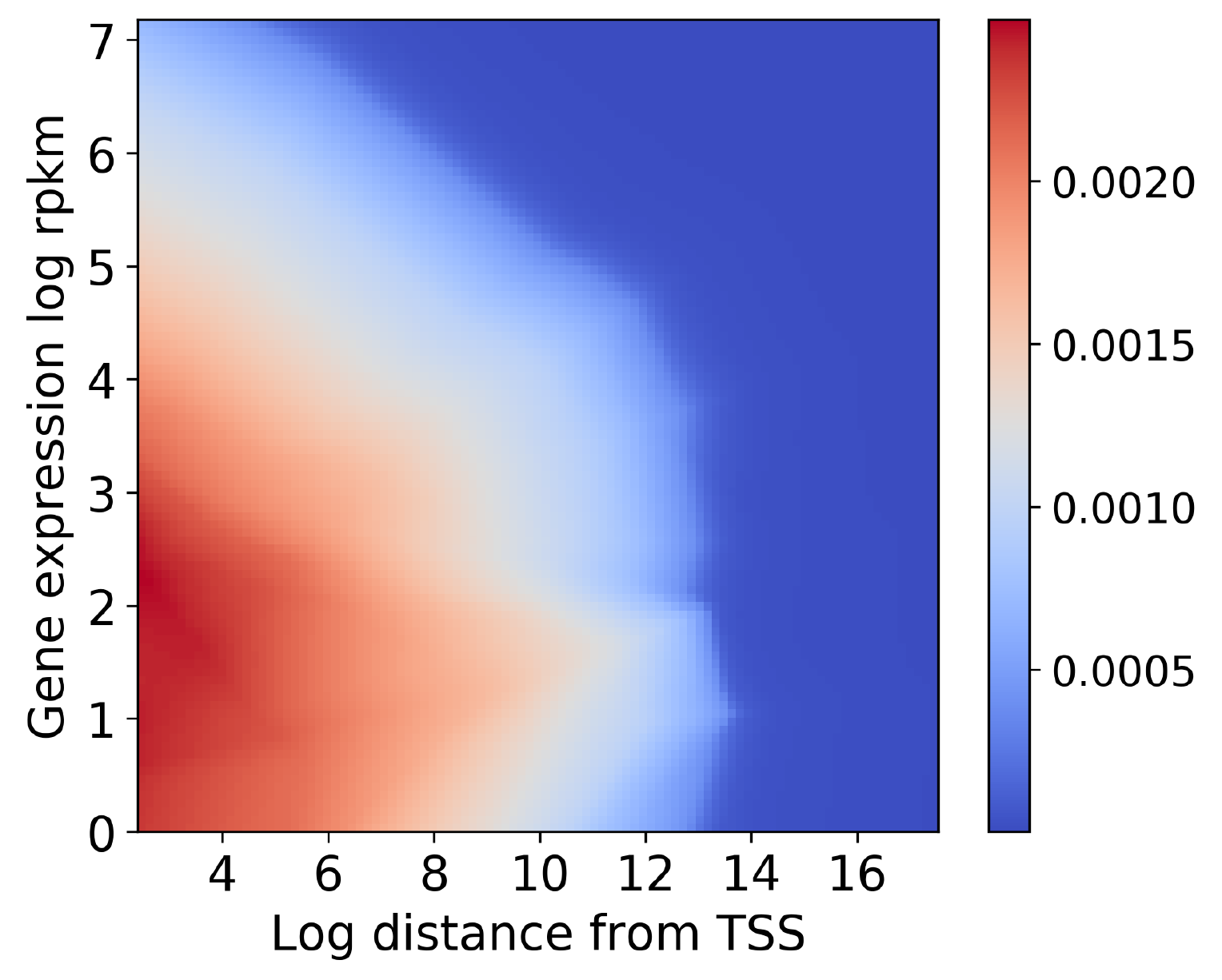}}
\subfigure[]{\includegraphics[width=0.3\linewidth]{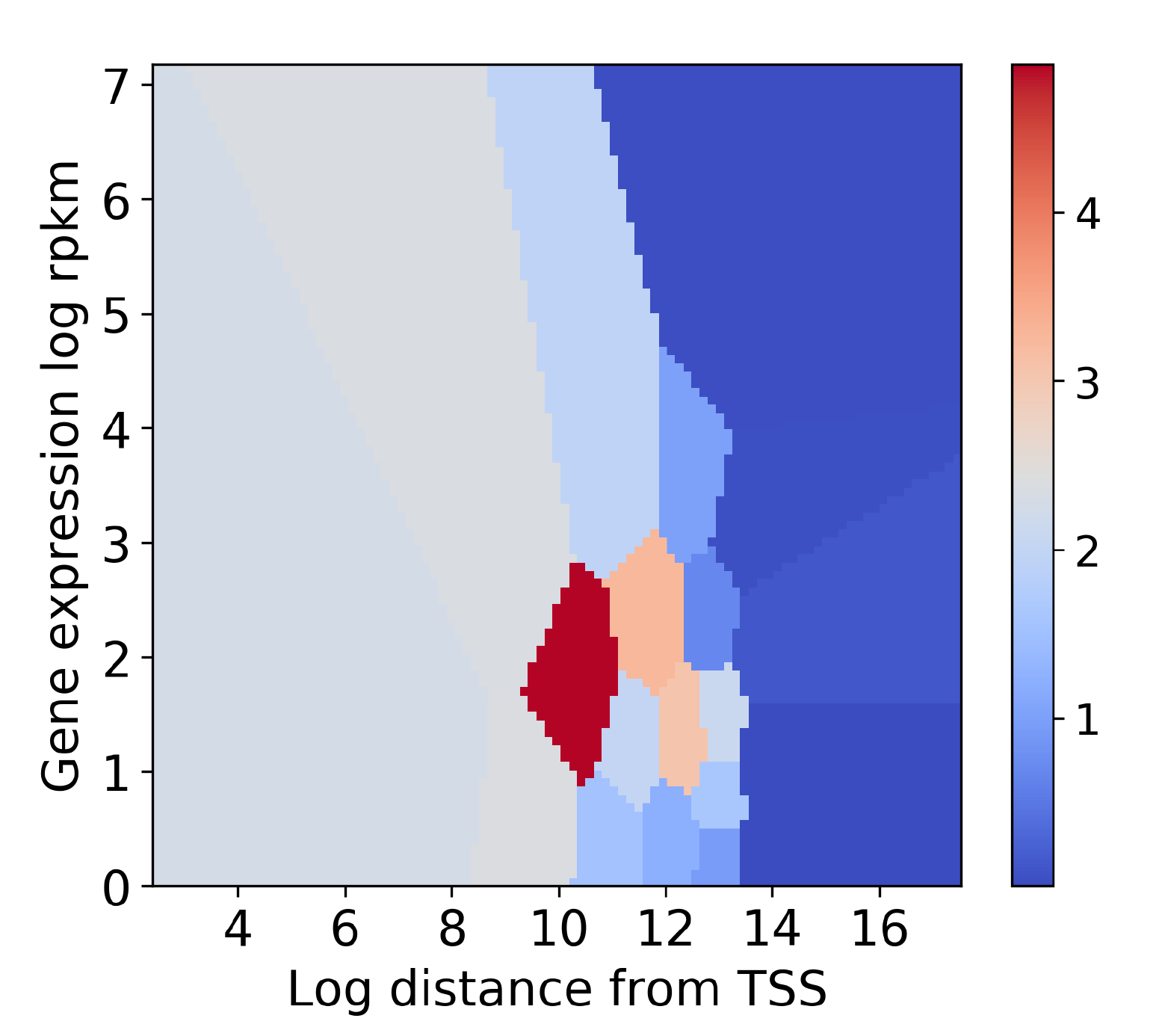}}
\subfigure[]{\includegraphics[width=0.33\linewidth]{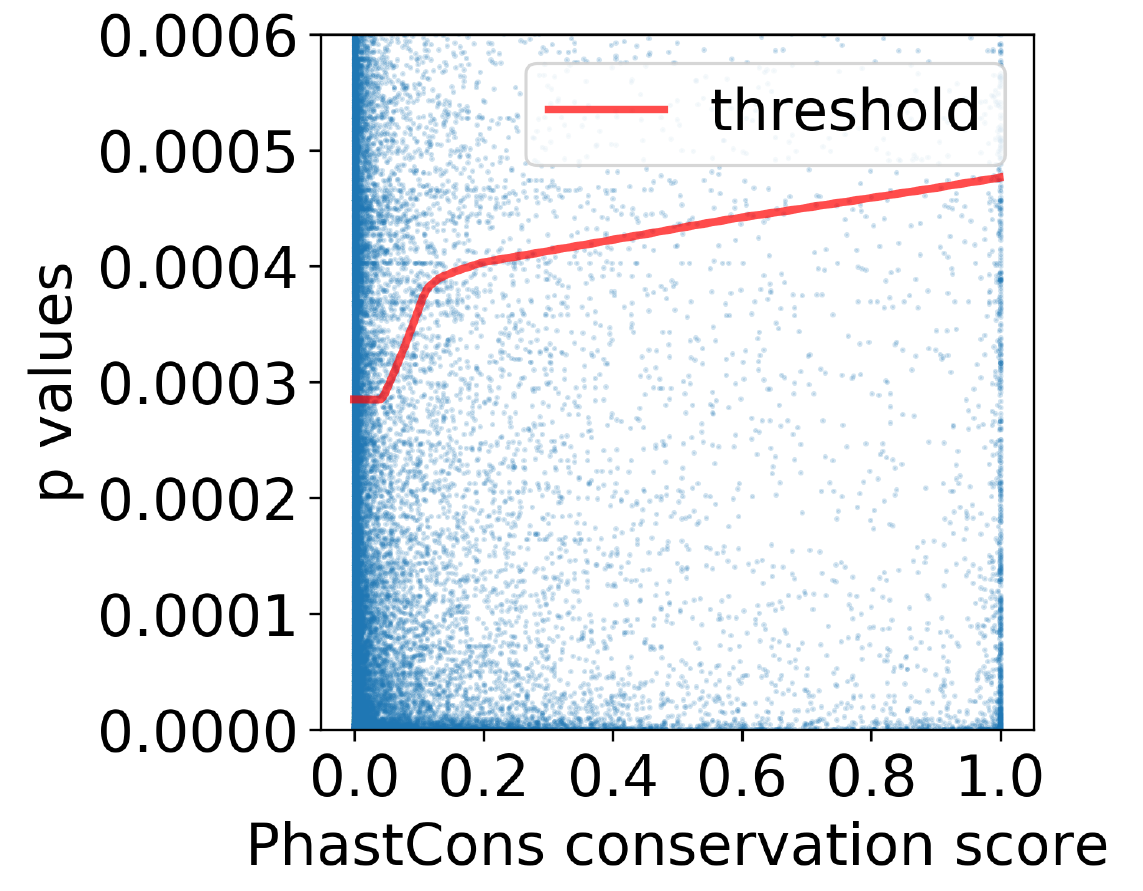}}
\caption{{\bf (a-c)} Results for 2DGM: (a) the alternative proportion for 2Dslope; (b) NeuralFDR's learned threshold; (c) IHW's learned weights.
{{\bf (d-e)} Results for GTEx-2D}: (d) NeuralFDR's learned threshold;(e) IHW's learned weights. (f) NeuralFDR's learned threshold for GTEx-PhastCons. 
\label{fig:2dgaussian}
%\james{Add GTEx conservation threshold plot?}
}
\end{figure}

\begin{figure}[htb]
\centering
\includegraphics[width=0.4\linewidth]{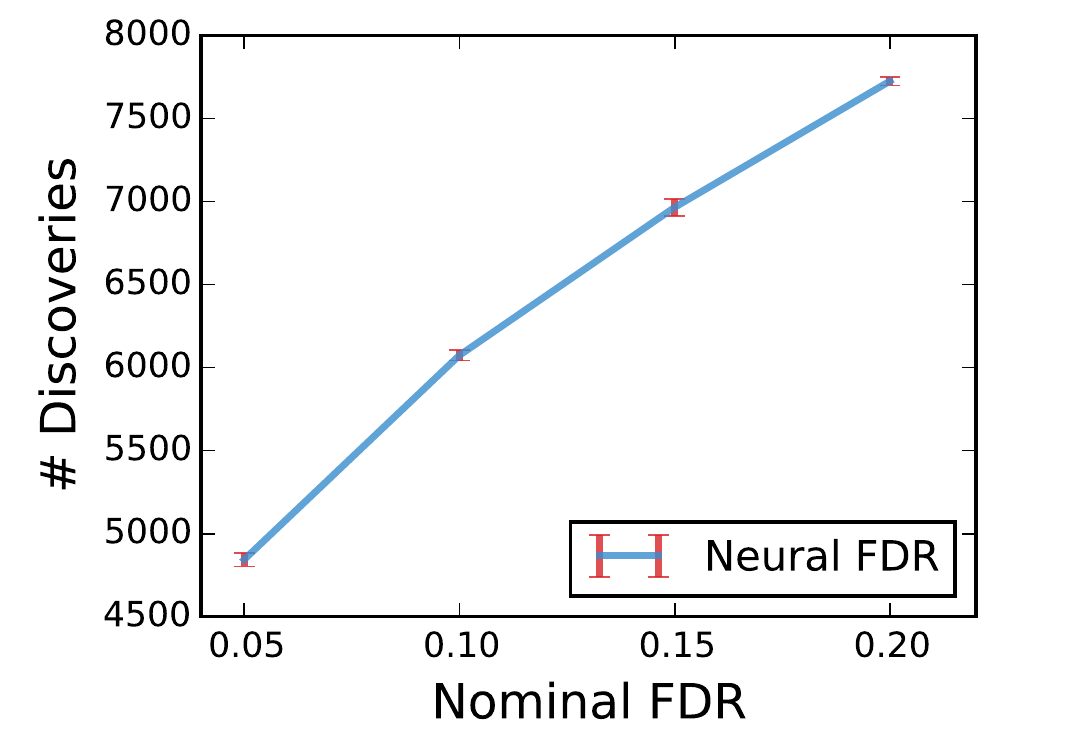}
\label{fig:stability}
\caption{Resutls of parallel runs for airway dataset, and it demonstrates the variation aross runs is small}
\end{figure}

As \texttt{NeuralFDR} uses neural network to do functional estimation, it has some randomness across mutiple runs. For example, the network could converge to bad local minimal. However, we show that \texttt{NeuralFDR} is stable across multiple runs. Fig. \ref{fig:stability} shows the number of discoveries in Airway dataset in 10 parallel runs for each nominal FDR. The errorbar denotes standard deviation, i.e. $68.3\%$ confidence interval. The coefficient of variation (CV) for each nominal FDR is smaller than $1\%$ across experiments.

\section{Implementation and Training Details \label{sec:S_imp}}

{\bf Objective function.} 
We solve the constrained optimization problem \ref{eq:opt_pblm} by the penalty method. We solve this optimization problem:

\begin{align} \label{eq:opt_pblm2}
\text{maximize}_{\bm{\theta}} \sum_{i} \tilde{D}(t(\bm{\theta})) - \lambda_1 \max \left\{ \widetilde{FD} (\gamma_i t^*_i (\bm{\theta})) - \alpha \tilde{D}(t(\bm{\theta}))  , 0 \right\}.
\end{align}

To avoid using step function, we used sigmoid to approximate the counting.
Denote the sigmoid function as $\sigma$.
We define $\tilde{D}$ and $\widetilde{FD}$ to be the following.
\begin{align}
    \tilde{D}(t(\bm(\theta))) &= \sum_j \sigma(\lambda_2 (t(\bm(\theta); x_j) - p_j) ) \\
    \widetilde{FD}(t(\bm(\theta))) &= \sum_j \sigma(\lambda_2 ( p_j  - (1-t(\bm(\theta); x_j)))) 
\end{align}

In cross validation process, we don't use approximated version of $D$ and $FD$. We use the actual number of points below the threshold and above the mirrored threshold as $D$ and $FD$.

{\bf Initialization.} 
As the optimization problem is highly non-convex, a good initialization is crucial for training. We used a smoothed version of $k$-mean clustering for initialization. The data is clustered into $k$ clusters using $k$-mean clustering based on the hypothesis features. An optimal threshold for each cluster $t_{opt,k}$ is calculated following Theorem \ref{thrm:opt_cond}. For each hypothesis, the initial value of the threshold is set to be 
\begin{align}
t_{init,i} = \sum_{j=1}^k  \frac{\exp(-\lambda_3||x_i - c_j||^2)}{\sum_{r=1}^n\exp(-\lambda_3||x_i - c_r||^2)} t_{opt,j}
\end{align}
where $c_j$ is the center for cluster $j$.

{\bf Network architecture.}
We used a 10-layer of MLP, each layer has 10 nodes. For activation function, we used LeakyReLU with a slope of 0.2. In the output layer, we use a scaled version of Sigmoid function to make sure the output is in $(0,0.5)$.

{\bf Implementation and Training.}
The algorithm is implemented in Python and the MLP is implemented using PyTorch. The optimization is solved using adaptive stochastic gradient method Adagrad \cite{duchi2011adaptive}. 

For all the experiments, we split the data equally into $M=3$ folds for cross validation. The learning rate is set to be $0.01.$ Because the optimization is driven by density, we use a large batch size of $10000$. Penalty parameter $\lambda_1$ is set to $20$, $\lambda_2$ is adaptively set depending on the BH threshold for a certain dataset, $\lambda_3$ is set to be $1$. All hyper-parameters are not heavily tuned and work across datasets. Training to fitting converges at around $6000$ iterations and for optimizing the number of discoveries converges at around $12000$ iterations. The training is done on Nvidia Tesla K80 GPUs. 

{\bf Notes for GTEx dataset.}
For GTEx dataset, the whole dataset is very large, so we filtered the p-values to get only hypothesis with $p < 0.005$ or $p > 0.995$, where the second part is for mirroring estimation. We also scale the network output to operate only in $[0, 0.005]$.

\section{Asymptotic FDR Control Under Weak Dependence \label{supp_sec:FDR_WD}}
Besides the i.i.d. case, \texttt{NeuralFDR} also controls FDR asymptotically under weak dependence \cite{storey2004strong,hu2010false}.
Extending the weak dependence definition in \cite{hu2010false} from discrete groups to continuous features $\mathbf{X}$, the data are defined to be under weak dependence if the CDF of $(P_i, X_i)$ for the null and the alternative proportion converge almost surely to their true values respectively.
The linkage disequilibrium (LD) in GWAS and the correlated genes in RNA-Seq can be addressed by such dependence structure.

\begin{definition} (Weak dependence\label{def:weak_dependence})
For the data $\{(P_i,\mathbf{X}_i,H_i)\}_{i=0}^n$ with the marginal distribution described by \eqref{eq:model}, let $F_0(p,\mathbf{x})$ and $F_1(p,\mathbf{x})$ be the cumulative density function of the distributions over $(P_i, \mathbf{X}_i)$ defined as $\P(P_i\leq p,\mathbf{X}_i\leq \mathbf{x},H_i=0)$, $\P(P_i\leq p,\mathbf{X}_i\leq \mathbf{x},H_i=1)$ respectively, where the inequality for vectors are element-wise.
The data is under weak dependence if $\forall (p,\mathbf{x})$,
\begin{align*}
\frac{1}{n} \sum_{i=1}^n\I_{\{P_i\leq p,\mathbf{X}_i\leq \mathbf{x}, H_i=0\}} \overset{a.s.}{\to} F_0(p,\mathbf{x}),~~~~\sum_{i=1}^n\I_{\{P_i\leq p,\mathbf{X}_i\leq \mathbf{x}, H_i=1\}} \overset{a.s.}{\to} F_1(p,\mathbf{x}).
\end{align*}
\end{definition}

\begin{theorem} (FDP control under weak dependence\label{thrm:FDP_ctrl_WD})
Under weak dependence, \texttt{NeuralFDR} with weight clipping controls FDR asymptotically. The weight clipping refers to clamping the weights to a bounded set after each gradient update when training the neural network \cite{arjovsky2017wasserstein}. 
\end{theorem}

\begin{proof} (Proof of theorem \ref{thrm:FDP_ctrl_WD}.)
Partition the space of $(p,\mathbf{x})$ into $k$ small boxes $B_1,\cdots,B_k$.
Under the weak dependence assumption Def. \ref{def:weak_dependence}, the proportion of elements in each box $B_j$, $\frac{1}{n}\sum_{i=1}^n \I_{\{(P_i,\mathbf{X}_i)\in B_j\}}$, converges uniformly to its true value, both for the CV set and the testing set. 
As $k\to \infty$, the boxes become smaller. Then for the family of Lipschitz continuous thresholds, their corresponding mirror estimates can be uniformly approximated by the proportion of elements in the boxes above the mirrored threshold.
Hence, as $k\to \infty$, the mirror estimates converge uniformly to their true values for the family of Lipschitz continuous thresholds.
Since \texttt{NeuralFDR} with weight clipping produces Lipschitz continuous thresholds, regardless of the value of $L$, the mirror estimates on the CV set and on the testing set will converge to their true values. 
Hence the difference of the mirror estimates on the CV set and on the testing set will converge to zero, giving that \texttt{NeuralFDR} controls FDR asymptotically.
\end{proof}

\begin{remark} (Lipschitz continuity)
In the i.i.d. case, we do not need Lipschitz continuity because for any $L$ learned thresholds based on the training data, their concentration on the CV data can be characterized by the concentration inequality and the union bound due to the i.i.d. structure. Therefore a FDR control guarantee can be established.
For the weakly dependent case, however, the convergence rate is hard to characterize. 
All we have is a point-wise almost surely convergence, with rate unknown. 
Hence, we first establish a uniform convergence for the $k$ boxes, or in other words simple functions with a fixed resolution depending on $k$, and use them to approximate the learned threshold. 
In this case, since we have no idea what the convergence rate of the $L$ learned thresholds will be like, we seek a uniform approximation of the family of learned thresholds.
Then this family should have some nice properties regarding the continuity in order for the approximation to be true, and Lipschitz continuity is one of the options.
\end{remark}

\section{Proofs of the Theoretical Results \label{sec:S_theo}}
\subsection{Proof of Lemma \ref{lm:bia_mir}}
\begin{proof} (Proof of Lemma \ref{lm:bia_mir})
$\forall \mathbf{x}$, given $\mathbf{X}_i=\mathbf{x}$ and $H_i=0$, $P_i \sim \text{Unif}(0,1)$. Then the joint distribution $f_{P \mathbf{X} H}(p, \mathbf{x}, 0)$ is also uniform w.r.t. $p$. We have  
\begin{align*}
\P((P_i,X_i) \in C(t), H_i=0) = \P((P_i,X_i) \in C^M(t), H_i=0).
\end{align*}

Then 
\begin{align*}
& \E[\widehat{FD}(t)] = \sum_{i=1}^n \P((P_i, X_i) \in C^M(t)) \\&= \sum_{i=1}^n \P((P_i, X_i) \in C^M(t), H_i=0) + \sum_{i=1}^n \P((P_i, X_i) \in C^M(t), H_i=1) \\ & 
= \sum_{i=1}^n  \P((P_i, X_i) \in C(t), H_i=0) + \sum_{i=1}^n \P((P_i, X_i) \in C^M(t), H_i=1) \\ 
&= \E[FD(t)] + \sum_{i=1}^n \P((P_i, X_i) \in C^M(t), H_i=1).
\end{align*}
\end{proof}

\subsection{Proof of Theorem \ref{thrm:FDP_ctrl}}
\begin{proof} (Proof of Theorem \ref{thrm:FDP_ctrl})
Consider fold $j$. 
Any decision rule candidate $t_{jl}$ may depend on the training set  $\mathcal{D}_{tr}(j)$ but is independent of the cross validation set $\mathcal{D}_{cv}(j)$. 
Thus assuming the first point is not in the training set, we can define
\begin{align*}
& p_{jl}= \frac{1}{m} \E[FD(t_{jl},\mathcal{D}_{cv}(j))]= \P(P_1 \leq t_{jl}(X_1), H_1=0)\\
& \overline{p}_{jl}= \frac{1}{m} \E[\widehat{FD}(t_{jl},\mathcal{D}_{cv}(j))]= \P(P_1 \geq 1-t_{jl}(X_1))\\
& q_{jl}= \frac{1}{m} \E [D(t_{jl},\mathcal{D}_{cv}(j))] = \P(P_1 \leq t_{jl}(X_1)).
\end{align*}
Notice that $p_{jl}\leq \overline{p}_{jl}$. 

Let the mirror estimate $\widehat{FD}(t_{jl},\mathcal{D}_{cv}(j))$ and the number of discoveries $D(t_{jl},\mathcal{D}_{cv}(j))$ be the quantities evaluated on $\mathcal{D}_{cv}(j)$. 
We know $\widehat{FD}(t_{jl},\mathcal{D}_{cv}(j)) \sim \text{Bin}(m,\overline{p}_{jl})$ and $D(t_{ij},\mathcal{D}_{cv}(j)) \sim \text{Bin}(m,q_{jl})$. By Lemma \ref{lm:chernoff}, 
\begin{align*}
&\P \left(  \widehat{FD}(t_{jl},\mathcal{D}_{cv}(j)) \leq (1-\delta_1) m \overline{p}_{jl} \right) < \exp -\frac{\delta_1^2 m\overline{p}_{jl}}{2},&\forall~0<\delta_1<1 \\
& \P \left(  D(t_{jl},\mathcal{D}_{cv}(j)) \geq (1+\delta_2) m q_{jl} \right) < \exp -\frac{\delta_2^2 m q_{jl}}{2+\delta_2},&\forall~\delta_2>0
\end{align*}
As $\widehat{FDP}(t_{jl},\mathcal{D}_{cv}(j)) = \frac{\widehat{FD}(t_{jl},\mathcal{D}_{cv}(j))}{D(t_{jl},\mathcal{D}_{cv}(j))}
$, we have
\begin{align*}
\P\left( \widehat{FDP}(t_{jl},\mathcal{D}_{cv}(j)) \leq \frac{1-\delta_1}{1+\delta_2} \frac{\overline{p}_{jl}}{q_{jl}} \right) < \exp -\frac{\delta_1^2 m\overline{p}_{jl}}{2} + \exp -\frac{\delta_2^2 m q_{jl}}{2+\delta_2}.
\end{align*} 
Consider the ``bad'' event rule $t_{jl}$ such that $\frac{\overline{p}_{jl}}{q_{jl}} \geq \frac{1+\delta_2}{1-\delta_1} \alpha$ or $q_{jl} \leq \frac{1}{1+\delta_2} c_0$, 
\begin{align*}
& \P\left( \widehat{FDP}(t_{jl},\mathcal{D}_{cv}(j)) \leq \alpha, \frac{D(t_{jl}, \mathcal{D}_{cv}(j))}{m} \geq c_0  \right) \\ 
& \leq \P\left( \widehat{FDP}(t_{jl},\mathcal{D}_{cv}(j)) \leq \alpha, \frac{\overline{p}_{jl}}{q_{jl}} \geq \frac{1+\delta_2}{1-\delta_1} \alpha  \right) \bigvee \P\left( \frac{D(t_{jl}, \mathcal{D}_{cv}(j))}{m} \geq c_0,q_{jl} \leq \frac{1}{1+\delta_2} c_0  \right)\\ 
&< \left( \exp -\frac{\delta_1^2 m \alpha q_{jl}}{2} + \exp -\frac{\delta_2^2 m q_{jl}}{2+\delta_2} \right) \bigvee \left(  \exp -\frac{\delta_2^2 m q_{jl}}{2+\delta_2}  \right) = \exp -\frac{\delta_1^2 m \alpha q_{jl}}{2} + \exp -\frac{\delta_2^2 m q_{jl}}{2+\delta_2},
\end{align*}
where for the second inequality we note that $\overline{p}_{jl}\geq \frac{1+\delta_2}{1-\delta_1} \alpha q_{jl}>\alpha q_{jl}$. 

Let $\underline{q} = \frac{1}{1+\delta_2} c_0$ and let $\mathcal{S} = \{t_{jl}: \frac{\overline{p}_{jl}}{q_{jl}} \geq \frac{1+\delta_2}{1-\delta_1} \alpha$ or $q_{jl} \leq \underline{q} \}$. We know that there are at most $L$ elements in $\mathcal{S}$.  
Then by the union bound,
\begin{align*}
\P\left(\exists~l\in\mathcal{S}, s.t.~~\widehat{FDP}(t_{jl},\mathcal{D}_{cv}(j)) \leq \alpha,   \frac{D(t_{jl}, \mathcal{D}_{cv}(j))}{m} \geq c_0 \right) < L \left(\exp -\frac{\delta_1^2 \alpha m  \underline{q}}{2} + \exp -\frac{\delta_2^2 m \underline{q}}{2+\delta_2} \right).
\end{align*}
Furthermore, let the $l^*$-th rule, $t_{jl^*}$, be the rule selected for testing. Note here $l^*$ is a random variable, $\widehat{FDP}(t_{jl^*}, \mathcal{D}_{cv}(j)) \leq \alpha$, and $\frac{D(t_{jl^*}, \mathcal{D}_{cv}(j))}{m} \geq c_0$. Therefore
\begin{align*}
\P(l^* \in \mathcal{S})=\P\left(\frac{\overline{p}_{jl^*}}{q_{jl*}} \geq \frac{1+\delta_2}{1-\delta_1} \alpha~~or~~q_{jl^*} \leq \underline{q} \right) < L \left(\exp -\frac{\delta_1^2 \alpha m  \underline{q}}{2} + \exp -\frac{\delta_2^2 m \underline{q}}{2+\delta_2} \right).
\end{align*}
As $p_{jl^*} < \overline{p}_{jl^*}$, we have 
\begin{align}\label{eq:FDRctrl5}
\P\left(\frac{p_{jl^*}}{q_{jl^*}} \geq \frac{1+\delta_2}{1-\delta_1} \alpha~~or~~q_{jl^*} \leq \underline{q} \right) < L \left(\exp -\frac{\delta_1^2 \alpha m  \underline{q}}{2} + \exp -\frac{\delta_2^2 m \underline{q}}{2+\delta_2} \right).
\end{align}

Now we move to the test data $\mathcal{D}_{te}(j)$. Again by Lemma \ref{lm:chernoff}, given $t_{jl^*}$,
\begin{align*}
&\P \left(  FD(t_{jl^*},\mathcal{D}_{te}(j)) \geq (1+\delta_3) m p_{jl^*} \vert t_{jl^*} \right) < \exp -\frac{\delta_3^2 m p_{jl^*}}{2+\delta_3},&\forall~\delta_3>0\\
& \P \left(  D(t^*_{j},\mathcal{D}_{te}(j)) \leq (1-\delta_4) m q_{jl^*} \vert t_{jl^*}\right) < \exp -\frac{\delta_4^2 m q_{jl^*}}{2},&\forall~0<\delta_4<1. 
\end{align*}
Then, given $t_j^*$,
\begin{align*}
\left.\P\left( FDP(t_{jl^*},\mathcal{D}_{te}(j)) \geq \frac{1+\delta_3}{1-\delta_4}\frac{p_{ij^*}}{q_{jl^*}} \right\vert t_{jl^*} \right) < \exp -\frac{\delta_3^2 mp_{jl^*}}{2+\delta_3} + \exp -\frac{\delta_4^2 m q_{jl^*}}{2}.
\end{align*}

The probability that FDP is large can be decomposed as follows:
\begin{align}\label{eq:FDRctrl1}
& \P\left( FDP(t_{jl^*},\mathcal{D}_{te}(j)) \geq \frac{1+\delta_2}{1-\delta_1} \frac{1+\delta_3}{1-\delta_4} \alpha \right) \\ 
& \leq \left.\P\left( FDP(t_{jl^*},\mathcal{D}_{te}(j)) \geq \frac{1+\delta_2}{1-\delta_1} \frac{1+\delta_3}{1-\delta_4} \alpha \right\vert \frac{p_{jl^*}}{q_{jl^*}} < \frac{1+\delta_2}{1-\delta_1} \alpha, q_{jl^*} \geq \underline{q}\right)\\&+ \P\left(\frac{p_{jl^*}}{q_{jl^*}} \geq \frac{1+\delta_2}{1-\delta_1} \alpha~~or~~q_{jl^*} \leq \underline{q} \right). 
\end{align}
For the conditional probability in the first term, we have 
\begin{align} \label{eq:FDRctrl4}
& \left.\P\left( FDP(t_{jl^*},\mathcal{D}_{te}(j)) \geq \frac{1+\delta_2}{1-\delta_1} \frac{1+\delta_3}{1-\delta_4} \alpha \right\vert \frac{p_{jl^*}}{q_{jl^*}} < \frac{1+\delta_2}{1-\delta_1} \alpha, q_{jl^*} \geq \underline{q} \right) \\
& \left. \leq \P\left( FDP(t^*_{j},\mathcal{D}_{te}(j)) \geq \frac{1+\delta_2}{1-\delta_1} \frac{1+\delta_3}{1-\delta_4} \alpha \right\vert  \frac{p_{jl^*}}{q_{jl^*}} = \frac{1+\delta_2}{1-\delta_1} \alpha, q_{jl^*} \geq \underline{q} \right) \\
& \leq \exp -\frac{\delta_3^2 \alpha m \underline{q}}{2+\delta_3} + \exp -\frac{\delta_4^2 m \underline{q}}{2}.
\end{align}
Combining \eqref{eq:FDRctrl5}, \eqref{eq:FDRctrl4}, \eqref{eq:FDRctrl1} can be written as 
\begin{align*}
& \P\left( FDP(t_{jl^*},\mathcal{D}_{te}(j)) \geq \frac{1+\delta_2}{1-\delta_1} \frac{1+\delta_3}{1-\delta_4} \alpha \right) \\
& < L \left(\exp -\frac{\delta_1^2 \alpha m  \underline{q}}{2} + \exp -\frac{\delta_2^2 m \underline{q}}{2+\delta_2} \right) + \left( \exp -\frac{\delta_3^2 \alpha m \underline{q}}{2+\delta_3} + \exp -\frac{\delta_4^2 m \underline{q}}{2}\right).
\end{align*}

Finally, by union bound over all $M$ folds,
\begin{align}\label{eq:FDRctrl2}
& \P\left(\exists~j,~FDP(t^*_{j},\mathcal{D}_{te}(j)) \geq  \frac{1+\delta_2}{1-\delta_1} \frac{1+\delta_3}{1-\delta_4} \alpha  \right) \\
&< LM \left(\exp -\frac{\delta_1^2 \alpha m  \underline{q}}{2} + \exp -\frac{\delta_2^2 m \underline{q}}{2+\delta_2} \right) + M \left( \exp -\frac{\delta_3^2 \alpha m \underline{q}}{2+\delta_3} + \exp -\frac{\delta_4^2 m \underline{q}}{2}\right),
\end{align}
for some $\delta_1, \delta_4 \in (0,1)$, $\delta_2, \delta_3 >0$. Note that \eqref{eq:FDRctrl2} also indicates that the overall FDP is smaller than $\frac{1+\delta_2}{1-\delta_1} \frac{1+\delta_3}{1-\delta_4} \alpha$.

Now let us derive an asymptotic bound when $\delta_1,\delta_2,\delta_3,\delta_4$ are close to $0$. In this case, we have $\delta_1,\delta_2,\delta_3,\delta_4 \in (0,1)$ and \eqref{eq:FDRctrl2} can be reduced to 
\begin{align}\label{eq:FDRctrl3}
& \P\left(\exists~j,~FDP(t^*_{j},\mathcal{D}_{te}(j)) \geq  \frac{1+\delta_2}{1-\delta_1} \frac{1+\delta_3}{1-\delta_4} \alpha  \right) \\
&< LM \left(\exp -\frac{\delta_1^2 \alpha m  \underline{q}}{2} + \exp -\frac{\delta_2^2 m \underline{q}}{3} \right) + M \left( \exp -\frac{\delta_3^2 \alpha m \underline{q}}{3} + \exp -\frac{\delta_4^2 m \underline{q}}{2}\right).
\end{align}
Let $\Delta = \min_j \delta_j$. Then $\frac{1+\delta_2}{1-\delta_1} \frac{1+\delta_3}{1-\delta_4}-1 = O(\Delta)$. 

For some $\beta >0$, let the four terms in \eqref{eq:FDRctrl3} be equal to $\frac{\beta}{4}$ so that the overall probability is $\beta$. This gives
\begin{align*}
\delta_1 = \sqrt{\frac{2}{\alpha m \underline{q}} \log \frac{4ML}{\beta}},~~~~
\delta_2 = \sqrt{\frac{3}{ m \underline{q}} \log \frac{4ML}{\beta}},~~~~
\delta_3 = \sqrt{\frac{3}{\alpha m\underline{q}} \log \frac{4M}{\beta}},~~~~
\delta_4 = \sqrt{\frac{2}{m\underline{q}} \log \frac{4M}{\beta}}.
\end{align*}
Thus $\Delta = \min_j \delta_j = O (\sqrt{\frac{M}{\alpha n } \log \frac{ML}{\beta}})$, where we note that the constant $\underline{q}$ is hidden inside the big $O$ term and $m=\frac{n}{M}$. This completes the proof.
\end{proof}

\subsection{Proof of Theorem \ref{thrm:opt_cond}}
\begin{proof} (Proof of Theorem \ref{thrm:opt_cond})
We first identify the worse case null proportion $\pi_0^*$. 
Consider any rule $t$. As $f_{P\mathbf{X}}$ is fixed, the probability of discovery $P_D(\gamma t,f_{P\mathbf{X}})$ is determined. 
For any two null proportions $\pi_0$ and $\pi_0'$, if $\forall \mathbf{x}, \pi_0(\mathbf{x}) \geq \pi_0'(\mathbf{x})$, the probability of false discovery $P_{FD}(t,f_{P\mathbf{X}}, \pi_0) \geq P_{FD}(t,f_{P\mathbf{X}}, \pi'_0)$, giving $FDP(t,f_{P\mathbf{X}}, \pi_0) \geq FDP(t,f_{P\mathbf{X}}, \pi'_0)$. Hence FDP is maximized when $\pi_0(\mathbf{x})$ is maximized for each point of $\mathbf{x}$. As $\forall \mathbf{x}$, $\pi_0(\mathbf{x})\leq f_{P\vert \mathbf{X}}(1\vert \mathbf{x})$ and $\pi_0^*(\mathbf{x})=f_{P\vert \mathbf{X}}(1\vert \mathbf{x})$ is attainable, we know for any rule $t$, FDP is maximized with $\pi_0^*(\mathbf{x})$. Then, problem \eqref{eq:opt} can be rewritten as 
\begin{align*}
\max_{t} P_D(t, f_{P\mathbf{X}})~~s.t.~~FDP(t,f_{P\mathbf{X}}, \pi^*_0) \leq \alpha.
\end{align*}

For condition (\ref{eq:opt_cond}.2), we prove by contradiction. Suppose $t$ is the optimal rule and $FDP(t,f_{P\mathbf{X}}, \pi^*_0)<\alpha$. Then there exists $\gamma>1$ such that $FDP(\gamma t,f_{P\mathbf{X}}, \pi^*_0) \leq \alpha$. As $P_D(\gamma t, f_{P\mathbf{X}}) > P_D( t, f_{P\mathbf{X}})$, $t$ can not be the optimal rule, giving the contradiction. 

For condition (\ref{eq:opt_cond}.1), we also prove by contradiction. Again suppose $t$ is the optimal rule where (\ref{eq:opt_cond}.1) is not met, then there exists $\mathcal{X}_1, \mathcal{X}_2 \subset \mathcal{X}$ with positive measure such that 
\begin{align*}
\frac{\int_{\mathcal{X}_1} f_{P \mathbf{X}}(1,\mathbf{x}) d\mathbf{x}}{\int_{\mathcal{X}_1} f_{P \mathbf{X}}(t(\mathbf{x}),\mathbf{x}) d\mathbf{x}}< \frac{\int_{\mathcal{X}_2} f_{P \mathbf{X}}(1,\mathbf{x}) d\mathbf{x}}{\int_{\mathcal{X}_2} f_{P \mathbf{X}}(t(\mathbf{x}),\mathbf{x}) d\mathbf{x}}. 
\end{align*}
Note that $f_{P \mathbf{X}}(p,\mathbf{x})$ is monotonically decreasing w.r.t. $p$. Then there exists $\epsilon>0$ such that for any $\epsilon_1, \epsilon_2 \in (0, \epsilon)$, 
\begin{align}\label{eq:thrm2_pf1}
\frac{\epsilon_1  \int_{\mathcal{X}_1}  f_{P \mathbf{X}}(1,\mathbf{x}) d\mathbf{x}}{\int_{\mathcal{X}_1}\int_{t(\mathbf{x})}^{t(\mathbf{x})+\epsilon_1} f_{P \mathbf{X}}(t(\mathbf{x})+\epsilon_1,\mathbf{x}) dp~ d\mathbf{x}}< \frac{\epsilon_2 \int_{\mathcal{X}_2} f_{P \mathbf{X}}(1,\mathbf{x}) d\mathbf{x}}{\int_{\mathcal{X}_2} \int_{t(\mathbf{x})-\epsilon_2}^{t(\mathbf{x})} f_{P \mathbf{X}}(t(\mathbf{x})-\epsilon_2,\mathbf{x}) dp~d\mathbf{x}}. 
\end{align}
Then we can pick $\epsilon_1, \epsilon_2 < \epsilon$ such that 
\begin{align}\label{eq:thrm2_pf2}
\int_{\mathcal{X}_1}\int_{t(\mathbf{x})}^{t(\mathbf{x})+\epsilon_1} f_{P \mathbf{X}}(p,\mathbf{x}) dp~ d\mathbf{x} = \int_{\mathcal{X}_2} \int_{t(\mathbf{x})-\epsilon_2}^{t(\mathbf{x})} f_{P \mathbf{X}}(p,\mathbf{x}) dp~d\mathbf{x} >0
\end{align}
Defining a new rule $t'(\mathbf{x})$ as 
\begin{align*}
t'(\mathbf{x}) = \left\{\begin{array}{cc}
t(\mathbf{x})+\epsilon_1, & \mathbf{x} \in \mathcal{X}_1 \\
t(\mathbf{x})-\epsilon_2, & \mathbf{x} \in \mathcal{X}_2 \\
t(\mathbf{x}),& otherwise
\end{array} \right..
\end{align*}
Then for the probability of discovery,
\begin{align*}
& P_D(t', f_{P\mathbf{X}}) = \int_{\mathcal{X}} \int_0^{t'(\mathbf{x})} f_{P \mathbf{X}}(p, \mathbf{x}) dp~d\mathbf{x} = \int_{\mathcal{X}} \int_0^{t(\mathbf{x})} f_{P \mathbf{X}}(p, \mathbf{x}) dp~d\mathbf{x} \\ &+ \int_{\mathcal{X}_1}\int_{t(\mathbf{x})}^{t(\mathbf{x})+\epsilon_1} f_{P \mathbf{X}}(p,\mathbf{x}) dp~ d\mathbf{x} - \int_{\mathcal{X}_2} \int_{t(\mathbf{x})-\epsilon_2}^{t(\mathbf{x})} f_{P \mathbf{X}}(p,\mathbf{x}) dp~d\mathbf{x}  = P_D(t, f_{P\mathbf{X}}).
\end{align*}
Moreover, from \eqref{eq:thrm2_pf1} and \eqref{eq:thrm2_pf2} we also know 
\begin{align*}
\epsilon_1  \int_{\mathcal{X}_1}  f_{P \mathbf{X}}(1,\mathbf{x}) d\mathbf{x} < \epsilon_2 \int_{\mathcal{X}_2} f_{P \mathbf{X}}(1,\mathbf{x}) d\mathbf{x}. 
\end{align*}
Then 
\begin{align*}
& P_{FD}(t', f_{P\mathbf{X}}, \pi_0^*) = \int_{\mathcal{X}} t'(\mathbf{x}) \pi_0^*(\mathbf{x}) \mu(\mathbf{x}) d\mathbf{x} = \int_{\mathcal{X}} t'(\mathbf{x}) f_{P\mathbf{X}}(1,\mathbf{x})d\mathbf{x} \\
& = \int_{\mathcal{X}} t(\mathbf{x}) f_{P\mathbf{X}}(1,\mathbf{x})d\mathbf{x} + \epsilon_1  \int_{\mathcal{X}_1}  f_{P \mathbf{X}}(1,\mathbf{x}) d\mathbf{x} - \epsilon_2 \int_{\mathcal{X}_2} f_{P \mathbf{X}}(1,\mathbf{x}) d\mathbf{x} & \\&< \int_{\mathcal{X}} t(\mathbf{x}) f_{P\mathbf{X}}(1,\mathbf{x})d\mathbf{x} = P_{FD}(t, f_{P\mathbf{X}}, \pi_0^*)
\end{align*}
Then $FDR(t', f_{P\mathbf{X}}, \pi_0^*) < FDR(t, f_{P\mathbf{X}}, \pi_0^*) = \alpha$. According to condition (\ref{eq:opt_cond}.2), $t'$ can not be the optimal rule. As $t$ and $t'$ are both feasible and have the same discovery probability, $t$ can not be the optimal rule either, giving the contradiction. 
\end{proof}

\subsection{Ancillary lemmas}
\begin{lemma}\label{lm:chernoff}
(Chernoff bound) For i.i.d. random variables $X_1,\cdots, X_n \in [0,1]$, let $X =\sum_{i=1}^n X_i$ and let $\mu = \E[X]$. Then
\begin{align*}
& \P ( X \geq (1+\delta) \mu) < \exp -\frac{\delta^2 \mu}{2+\delta},&\forall~\delta>0 \\ 
& \P ( X \leq (1-\delta) \mu) < \exp -\frac{\delta^2 \mu}{2},&\forall~0<\delta<1.
\end{align*}
\end{lemma}

\clearpage
\begin{center}
\textbf{\large More Supplemental Notes }
\end{center}
\setcounter{section}{0}

\section{Improving the Mirror Estimator}
In a discussion with Nikos Ignatiadis, he pointed out that the mirror estimator can be substituted by the following estimator
\begin{align*}
\widehat{FD} (t) = \sum_{i=1}^n t(\mathbf{X}_i).
\end{align*}

This estimator also overestimates the expected value of the false discoveries. In addition, the randomness of this estimator is over only $\mathbf{X}_i$. This should potentially make the estimate more stable. 

The overestimation property can be shown as follows.

\begin{align*}
& \frac{1}{n}\E[\widehat{FD} (t)] \overset{\mathbf{X}_i~i.i.d.}{=}  \E [ t(\mathbf{X}_1) ] \overset{tower~property}{=}  \E [ \E [t(\mathbf{x}_1)\vert \mathbf{X}_1=\mathbf{x}_1]  ] \\ & \overset{null~distribution~uniform}{=} \E [ \E [\P[P_1 \leq t(\mathbf{x}_1) \vert H_1=0] \vert \mathbf{X}_1=\mathbf{x}_1]  ] 
\\& = \E \left[ \E \left[ \frac{\P(P_1 \leq t(\mathbf{x}_1), H_1=0)}{ \P(H_1=0) } \middle| \mathbf{X}_1=\mathbf{x}_1\right]  \right] 
% \\& = \E \left[ \E \left[ \frac{\P(P_1 \leq t(\mathbf{x}_1), H_1=0)}{ \pi_0(\mathbf{x}_1) } \middle| \mathbf{X}_1=\mathbf{x}_1\right]  \right] 
\\& \geq \E \left[ \E \left[ \P(P_1 \leq t(\mathbf{x}_1), H_1=0) \middle| \mathbf{X}_1=\mathbf{x}_1\right]  \right] = \P(P_1 \leq t(\mathbf{X}_1), H_1=0) = \frac{1}{n} \E[FD(t)].
\end{align*}

% Being i.i.d.  
% \begin{align*}
% &\text{maximize}_t~D(t)\\
% & s.t.~\widehat{FDP}(t) \leq \alpha 
% \end{align*}

% Conditional on $\mathbf{X}_i = \mathbf{x}_i$, we have 
% \begin{align*}
% & \frac{1}{n} \widehat{FD} (t) \big| \mathbf{X}_i = \mathbf{x}_i= \frac{1}{n} \sum_{i=1}^n t(\mathbf{x}_i) = \frac{1}{n} \sum_{i=1}^n \P(P_i \vert H_i=0) \big| \mathbf{X}_i = \mathbf{x}_i 
%  \\&= \frac{1}{n} \sum_{i=1}^n \frac{\P(P_i \leq t(\mathbf{x}_i, H_i=0))}{\P(H_i=0)} \big| \mathbf{X}_i = \mathbf{x}_i 
% \end{align*}

% \begin{align*}
% & \frac{1}{n}\E_{\mathbf{X}}[\widehat{FD} (t)] \overset{\mathbf{X}_i~i.i.d.}{=}  \E_{\mathbf{X}} [ t(\mathbf{X}_1) ]  \overset{uniform~null}{=} \E_{\mathbf{X}} [ \P_P[P_1 \leq t(\mathbf{X}_1) \vert H_1=0]] 
% \\& = \E_{\mathbf{X}} \left[ \frac{\P_{PH}(P_1 \leq t(\mathbf{X}_1), H_1=0)}{ \P_{H}(H_1=0) }   \right]  \geq \E_{\mathbf{X}} \left[ \P_{PH}(P_1 \leq t(\mathbf{X}_1), H_1=0)  \right]  \\&= \P(P_1 \leq t(\mathbf{X}_1), H_1=0) = \frac{1}{n} \E[FD(t)].
% \end{align*}

% On the 

\end{document}